\documentclass{aptpub}
\authornames{Predrag R. Jelenkovi\'c,    Evangelia D. Skiani} 
\shorttitle{Distribution of the Retransmissions of Bounded Documents} 

\usepackage{amsmath}
\usepackage{mathrsfs}
\usepackage{latexsym}
\usepackage{amssymb}
\usepackage{amsfonts}
\usepackage{pifont}
\usepackage{epsfig}
\usepackage{theorem}
\usepackage{array}

\def\qed{\hfill{$\Box$} \\}

\def \b{b}
\def \eqdef{\triangleq}

\def \E{{\mathbb  E}}
\def \Pr{{\mathbb P}}
\def \Phin{\Phi^{\leftarrow}}
\numberwithin{equation}{section}  
\newcommand{\mynote}[1]{\footnote{\hspace*{-14pt} #1}\par}

\begin{document}

\title{Distribution of the Number of Retransmissions of Bounded Documents} 

\authorone[Columbia University]{Predrag R. Jelenkovi\'c } 
\authorone[Columbia University]{Evangelia D. Skiani}

\mynote{\small This work is supported by the National Science Foundation grant number 0915784. }
\mynote{\small Extended abstract of this paper with preliminary results has appeared earlier in \cite{JS2012}.} 
\addressone{\small Department of Electrical Engineering, Columbia University, New York, NY 10027.}

\begin{abstract}
Retransmission-based failure recovery represents a primary approach in existing communication networks that guarantees data delivery in the presence of channel failures. Recent work has shown that, when data sizes have infinite support, retransmissions can cause long (-tailed) delays even if all traffic and network characteristics are light-tailed. In this paper we investigate the practically important case of bounded data units $0\le L_{\b}\le b$ under the condition that the hazard functions of the distributions of data sizes and channel statistics are proportional. To this end, we provide an explicit and uniform characterization of the entire body of the retransmission distribution $\Pr[N_{\b} > n]$ in both $n$ and $b$. Our main discovery is that this distribution can be represented as the product of a power law and Gamma distribution. This rigorous approximation clearly demonstrates the coupling of a power law distribution, dominating the main body, and the Gamma distribution, determining the exponential tail. Our results are validated via simulation experiments and can be useful for designing retransmission-based systems with the required performance characteristics. From a broader perspective, this study applies to any other system, e.g., computing, where restart mechanisms are employed after a job processing failure. 
\end{abstract}


\section{Introduction}
Failure recovery mechanisms are employed in almost all engineering networks 
since complex systems of any kind are often prone to failures.
One of the most straightforward and widely used failure recovery mechanism is to simply restart the system and all of the interrupted jobs from the beginning after a failure occurs. It was first recognized in \cite{FS05,SL06} that such mechanisms may result in long-tailed (power law) delays even if the job sizes and failure rates are exponential. In \cite{PJ07RETRANS}, it was noted that the same mechanism is at the core of modern communication networks where retransmissions are used on all protocol layers to guarantee data delivery in the presence of channel failures. Furthermore,  \cite{PJ07RETRANS} shows that the power law number of retransmissions and delay occur whenever the hazard functions of the data and failure distributions are proportional. Hence, power laws may arise even if the data and channel failure distributions are both Gaussian. In particular, retransmission phenomena can lead to zero throughput and system instabilities, 
and therefore need to be carefully considered for the design of fault tolerant systems.

More specifically, in communication networks, retransmissions represent the basic building blocks for failure recovery in all network protocols that guarantee data delivery in the presence of channel failures. These types of mechanisms have been employed on all networking layers, including, for example, Automatic Repeat reQuest (ARQ) protocol (e.g., see Section~2.4 of \cite{BG92}) in the
data link layer where a packet is resent automatically in case of an error;  contention based ALOHA type protocols in the medium access control (MAC) layer that use random backoff and retransmission mechanism to recover data from collisions; end-to-end acknowledgment for  multi-hop transmissions in the transport layer; HTTP downloading scheme in the application layer, etc. It has been shown that several well-known retransmission based protocols in different
layers of networking architecture  can lead to power law delays,
e.g., ALOHA type protocols in MAC layer \cite{Jelen07ALOHA, JT09} and
end-to-end acknowledgments in transport layer \cite{Jelen07e2e,JT2007}
as well as in other layers \cite{PJ07RETRANS}. For other (non-retransmission) mechanisms that can give rise to heavy tails see \cite{JT10} and the references therein. In particular, the proportional growth/multiplicative models can result in heavy tails \cite{JT10, JO12}.

Traditionally, retransmissions were thought to follow light-tailed distributions (with rapidly decaying tails), namely geometric, which requires the further assumption of independence between data sizes and transmission error probability. However, these two are often highly correlated in most communication systems, meaning that longer data units have higher probability of error, thus violating the independence assumption. Recent work \cite{PJ07RETRANS,Jelen07ALOHA,Jelen07e2e,JT2007} has  shown that, when the data size distribution has infinite support, all retransmission-based protocols could cause heavy-tailed behavior and possibly result in zero throughput, regardless of how light-tailed the distributions of data sizes and channel failures are. Nevertheless, in reality, data sizes are usually upper bounded. For example, WaveLAN's maximum transfer unit is 1500 bytes, YouTube videos are of limited duration, e-mail attachments cannot exceed an upper limit, say 25MB, etc. This fact motivates us to investigate the transmission of bounded data and approximate uniformly the entire body of the resulting retransmission distribution as it transits from the power law to the exponential tail. 

We use the following generic channel with failures \cite{PJ07RETRANS} to model the preceding situations. 
This model was first introduced in \cite{FS05} in a different application context. 
The channel dynamics is described by the i.i.d. channel availability process $\{A$, $A_i\}_{i\geq 1}$, where the channel is continuously available during periods $\{A_i\}$ and fails between these periods. In each period of time that the channel becomes available, say $A_i$, we attempt to transmit the data unit of random size $L_b$. We focus on the situation when the data size has finite support on interval $[0,\b]$. 
If $L_b< A_i$, we say that the transmission is successful; otherwise, we wait for the next period
$A_{i+1}$ when the channel is available and attempt to retransmit
the data from the beginning. It was first recognized in
\cite{FS05} that this model results in power law distributions when
the distributions of $L\equiv L_\infty$ and $A$ have a matrix exponential
representation, and this result was rigorously proved and further
generalized in \cite{PJ07RETRANS,JT2007,Asmussen07}. A related study when $L = \ell$ is a constant and failure/arrival rates are time-dependent Poisson can be found in \cite{Asmussen10}.

It was discovered in \cite{PJ07RETRANS} that bounded data units result in truncated power law distributions for the number of retransmissions, see Example~3 in \cite{PJ07RETRANS}; see also Example~2 in \cite{Jelen07ALOHA}. Such distributions are characterized by a power law main body and an exponentially bounded tail. However, the exponential behavior appears only for very small probabilities, often meaning that the number of retransmissions of interest may fall inside the region where the distribution behaves as a power law. 
It was argued in Example~3 of \cite{PJ07RETRANS} that the power law region will grow faster than 
exponential if the distributions of $A$ and $L_b$ are lighter than exponential. 
The retransmissions of bounded documents were further studied in \cite{TS2010}, where
partial approximations of the distribution of the number of retransmissions on the logarithmic and exact scales 
were provided in Theorems~1 and 3 of \cite{TS2010}, respectively. In this paper, we present a uniform characterization
of the entire body of such a distribution, both on the logarithmic as well as the exact scale.

Specifically, let $N_b$ represent the number of retransmissions (until successful transmission) of a bounded random data unit of size $L_b \in [0,\b]$ on the previously described channel. In order to study the uniform approximation in both $n$ and $\b$ we construct a family of variables $L_\b$, such that $\Pr[L_\b \leq x] = \Pr[L \leq x] / \Pr[L \leq \b]$, for $0 \leq x \leq \b$ when $L = L_\infty$ is fixed. This scaling of $L_\b$ was also used in \cite{TS2010}.
For the logarithmic scale, our result, stated in Theorem~\ref{thm:2},
provides a uniform characterization of the entire body of $\log \Pr[N_b >n]$,
i.e., informally 
\[
\log \Pr[N_{\b}>n] \approx -\alpha \log n + n \log \Pr[A \leq \b]
\]
for all $n$ and $b$ sufficiently large when the hazard functions of $L$ and $A$ are linearly related as $\log \Pr[L>x] \approx \alpha \log  \Pr[A > x]$; see Theorem~\ref{thm:2} for the precise assumptions. Note that the first term in the preceding approximation corresponds to the 
power law part $n^{-\alpha}$ of the distribution, while the second part describes the exponential 
(geometric $ \Pr[A \leq \b]^n$) tail. Hence, it may be natural to define the transition point $n_b$ from the 
power law to the exponential tail as a solution to $n_\b \log \Pr[A \leq \b] \approx \alpha \log n_\b$.

In addition, under more restrictive assumptions, we discover a new exact 
asymptotic formula for the retransmission distribution that works uniformly for all large $n,b$. 
Surprisingly, the approximation admits an explicit form (see Theorems~\ref{thm:3} and~\ref{thm:4})
\begin{align} \label{eq:0}
\Pr[N_\b >n] \approx \frac{\alpha} {n^\alpha \ell(n \wedge \Pr[A>\b]^{-1})} \int_{-n \log \Pr[A \leq \b]}^{\infty} e^{-z} z^{\alpha-1} dz ,
\end{align}
where $x \wedge y = \min (x,y)$ and $\ell(\cdot)$ is a slowly varying function; note that the preceding integral is the incomplete Gamma function $\Gamma(x,\alpha)$.

Clearly, when $-n \log \Pr[A \leq b] \downarrow 0$, the
preceding approximation converges to a true power law $\Gamma(\alpha+1)/(\ell(n) n^\alpha)$. And, when $-n  \log(\Pr[A<b])\uparrow \infty$, approximation 
\eqref{eq:0}, by the property $\Gamma(x,\alpha) \sim e^{-x} x^{\alpha-1}$ as $x\rightarrow\infty$, 
has a geometric leading term $\Pr[A \leq \b]^n$.
Interestingly, for the special case when $\alpha$ is an integer and $\ell(x) \equiv 1$, one can compute the exact expression for $\Pr[N_b>n]$, see Proposition~\ref{prop:integer}. 
Furthermore, our results show that the length of the power law region increases as the corresponding distributions of $L$ and $A$ assume lighter tails.  
All of the preceding results are validated via simulation experiments in Section~\ref{s:sim}.
It is worth noting that our asymptotic approximations are in excellent agreement with the simulations.

This uniform approximation allows for a characterization of the entire body of the distribution $\Pr[N_\b >n]$, so that one can explicitly estimate the region where the power law phenomenon arises. Introducing the relationship between $n$ and $\Pr [A > \b]$ also provides an assessment method of efficiency and is important for diminishing the power law effects in order to achieve high throughput. 
Basically, when the power law region is significant, it could lead to nearly zero throughput  ($\alpha < 1$), implying that the system parameters should be more carefully adjusted in order to meet the new requirements. On the contrary, if the exponential tail dominates, the system performance is more desirable. Our analytical work could be applicable in network protocol design, 
possibly including data fragmentation techniques \cite{Jelen07frag, NL2009} and failure-recovery mechanisms.

Also, from an engineering perspective, our results further suggest that careful re-examination and possible redesign of retransmission based protocols in communication networks might be necessary. Specifically, current engineering trends towards infrastructure-less, error-prone wireless technology encourage the study of highly variable systems with frequent failures. In these types of systems, traditional approaches, e.g., blind data fragmentation, may be insufficient for achieving a good balance between throughput and resource utilization. For example, IP packets are lower bounded by the packet header of 20 bytes and cannot be more than 1500 bytes. Thus, it is not efficient to create very small packets since the 20-byte packet header carries no useful information. In fact, one may consider merging smaller packets to reduce the overhead and, hence, increase the efficiency. Overall, we consider a generic model when the maximum size of data units is limited, which, in general, can be used towards improving the design of future complex and failure-prone systems in many different applications.

The rest of the paper is organized as follows. After a detailed
description of the channel model in the next Subsection
\ref{s:model}, we present our main results in Section
\ref{s:results}. Then, Section~\ref{s:sim} contains simulation examples that verify our theoretical work. A number of technical proofs are postponed to Section~\ref{sec:4}.

\subsection{Description of the Channel}\label{s:model}
In this section, we formally describe our model and provide necessary definitions and notation.
Consider transmitting a generic data unit of random size $L_\b$ over a
channel with failures. Without loss of generality, we assume that
the channel is of unit capacity. As stated in the introduction, the channel dynamics is modeled by the channel availability process $\{A, A_i\}_{i\geq 1}$, where the channel is continuously available during time periods $\{A_i\}$ whereas it fails between such periods. In each period of time that the channel becomes
available, say $A_i$, we attempt to transmit the data unit and, if
$A_i > L_\b$, we say that the transmission was successful; otherwise,
we wait for the next period $A_{i+1}$ when the channel is available
and attempt to retransmit the data from the beginning.
  A sketch of the
  model depicting the system is drawn in Figure \ref{fig:tc}.
 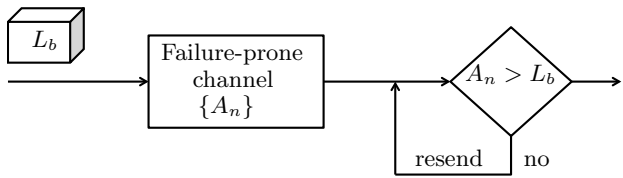
\begin{figure}[ht]
\centering
\begin{picture}(220,70)(0,0)
\put(0,0){\includegraphics[scale = 0.4, trim = 10mm 70mm 30mm 50mm, clip ]{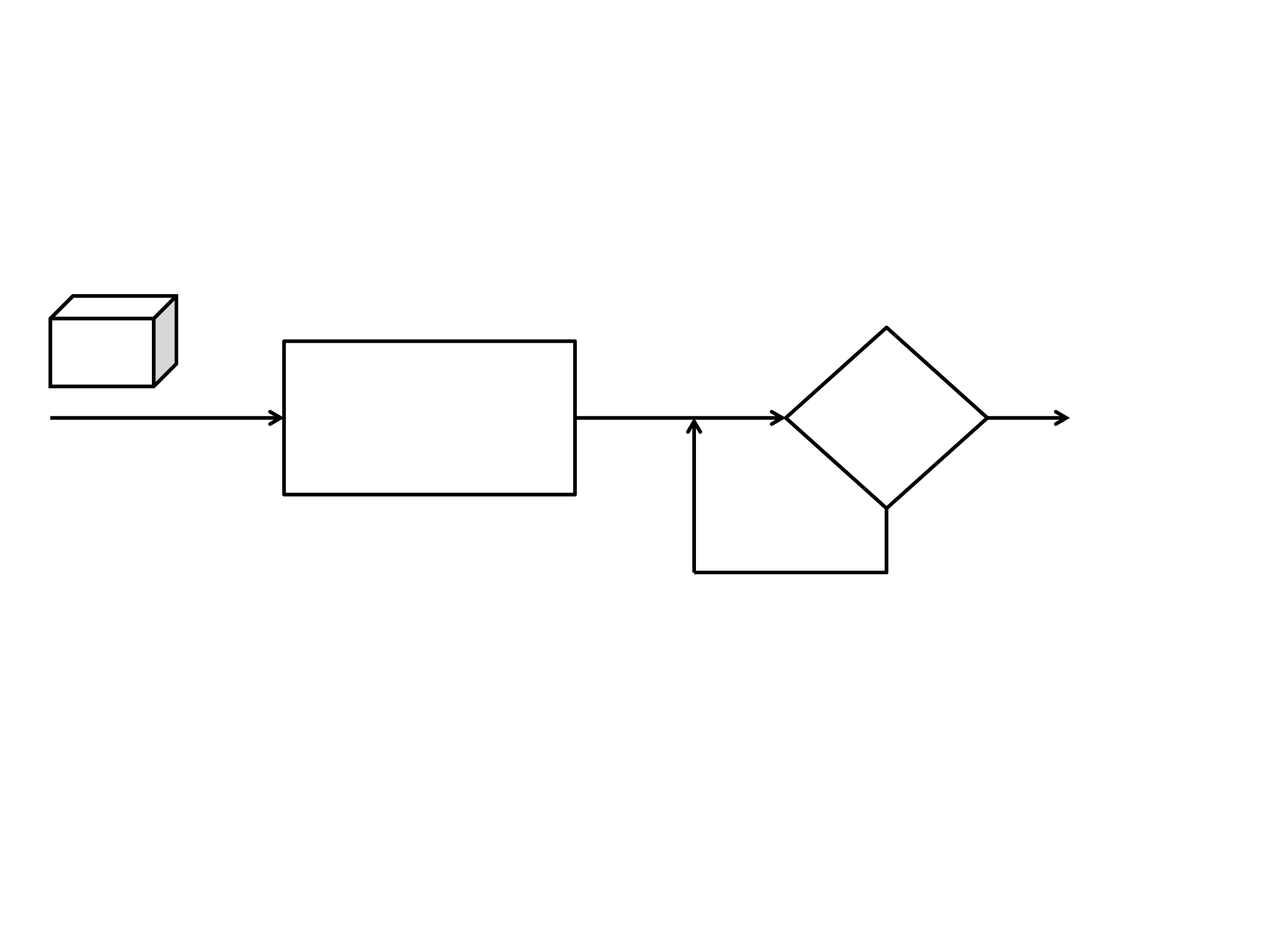}}
\put(9,55){\small $L_b$}
\put(58,50){\small Failure-prone }
\put(70,40){\small channel }
\put(72,30){\small $\{A_n\}$ }
\put(173,42){\small $A_n > L_b$}
\put(154,10){\small resend   $\quad$      no}
\end{picture}
  \caption{Documents sent over a channel with failures}
\label{fig:tc}
\end{figure}

We are interested in computing the number of attempts $N_\b$ (retransmissions) that is required until $L_\b$ is successfully transmitted, which is formally defined as follows. 
 \begin{definition}\label{def:NT}
 The total number of retransmissions for a generic data unit of length $L_\b$ is defined
 as
 \begin{equation*}
  N_\b \eqdef  \inf \{ n: A_n> L_\b \}.
  \end{equation*}
 \end{definition}
 We  denote
 the complementary cumulative distribution functions for $A$ and
 $L$, respectively, as
  \begin{equation*}
  \bar{G}(x)\eqdef \Pr[A>x]  \quad \text{and} \quad
  \bar{F}(x)\eqdef \Pr[L>x],
  \end{equation*}
  where $L$ is a generic random variable that is used to define the distribution of $L_\b$.

Throughout the paper we assume that $L$ and $A$ are continuous (equivalently, $\bar{F}(x)$ and $\bar{G}(x)$ are absolutely continuous) and have infinite support, i.e., $\bar{G}(x) >0$ and $\bar{F}(x) >0$ for all $x \ge 0$. Then, the distribution of $L_\b$ is defined as

\begin{equation} \label{eq:defF}
\Pr[L_\b  \leq x ] = \frac{\Pr[L \leq x]}{\Pr[L \leq b]}, \quad 0 \leq x \leq \b.
\end{equation}
To avoid trivialities, we assume that $\b$ is large enough such that $\Pr[L \leq \b] >0$.

In this paper we use the following standard notations. For any two real functions $a(t)$ and $b(t)$ and fixed $t_0 \in \mathbb{R} \bigcup \{\infty \}$, we use 
$a(t)\sim  b(t) \text{ as } t \rightarrow t_0 $ to denote $\lim_{t \rightarrow t_0} a(t)/b(t) =1$. Similarly, we say that $a(t) \gtrsim b(t)$ as $t \rightarrow  t_0$ if $\lim \inf _{t \rightarrow t_0}$ $ [a(t)/b(t)] \geq1$; $a(t) \lesssim b(t)$ has a complementary definition.

\section{Main Results} \label{s:results}
In this section, we present our main results. Under mild conditions, we first prove a general upper bound for the distribution of $N_\b$ on the logarithmic scale in Proposition~\ref{prop:UB}. 
In Theorem~\ref{thm:2}, we present our first main result, which under more stringent assumptions, characterizes the entire body of the distribution on the logarithmic scale uniformly for all large $n$ and $\b$, i.e., informally we show that $$\log \Pr[N_\b >n] \approx - \alpha \log n + n \log \Pr[A \leq b],$$as previously mentioned in the introduction. Roughly speaking, when $-\log \Pr[A \leq \b] = o(\log n/ n)$, $\Pr[N_\b >n]$ is a power law of index $\alpha$. Our results on the exact asymptotics are given in the next Subsection~\ref{sec:exact} in Theorems~\ref{thm:3} and~\ref{thm:4}; the results are stated in two different theorems since Theorem~\ref{thm:4} requires slightly stronger assumptions. The uniform approximation implied by these two theorems is presented in \eqref{eq:approx}, or previously in \eqref{eq:0}. 

Recall that the distribution of $L_{\b}$ has finite support on $[0,b]$, given by \eqref{eq:defF}. First, we prove the following general upper bound. 

\begin{proposition}\label{prop:UB}
Assume that 
\begin{displaymath}
\liminf_{x \rightarrow \infty} \frac { \log \Pr[L > x ] }{  \log  \Pr[A  > x ]} \geq \alpha
\end{displaymath}
and let $b_0$ be such that $\Pr[L \leq b_0] >0, \Pr[A \leq \b_0]>0$, then for any $\epsilon > 0$, there exists $n_0$, such that, for all $ n \geq n_0, \b \geq \b_0  $,
$$\log \Pr[N_\b > n]  \leq   (1-\epsilon) \left[ n \log \Pr[A \leq \b] - \alpha \log n \right].$$
\end{proposition}

\begin{remark}
Note that this result can be restated as 
\begin{displaymath}
\Pr[N_\b > n]  \leq   \Pr[A \leq \b]^{n(1-\epsilon)} n^{-\alpha(1-\epsilon)},
\end{displaymath}
for $n, \b$ sufficiently large. Hence, the distribution $\Pr[N_\b > n]$ is bounded by the product of a power law and a geometric term. 
\end{remark}

\begin{proof} By assumption, there exists $0<\epsilon <1 $ such that for all $x>x_{\epsilon} \geq \b_0 > 0,$
\begin{equation}\label{eq:prop1} 
 \bar{F}(x) \leq \bar{G}(x)^{\alpha(1-\epsilon)}. 
\end{equation}

\noindent Next, it is easy to see that $\Pr[N_\b >n |L_{\b} ] = (1-\bar{G}(L_{\b}))^n $, and thus,
\begin{align*}
\Pr[N_\b > n]&= \E [1-\bar{G}(L_\b)]^n \\
& =  \E [1-\bar{G}(L_\b)]^{n(1-\epsilon + \epsilon)} \\
& \leq  (1-\bar{G}(\b))^{n(1-\epsilon)}  \left[ \E[1-\bar{G}(L_\b)]^{n \epsilon} 1(L_\b \leq x_\epsilon) +  \E[1-\bar{G}(L_\b)]^{n \epsilon}1(L_\b > x_\epsilon) \right] \\
& \leq  (1-\bar{G}(\b))^{n(1-\epsilon)} \left[ (1-\bar{G}(x_\epsilon))^{n\epsilon} + \int_{x_\epsilon}^\b{\left(1-\bar{G}(x)\right)^{n \epsilon}  \frac{dF(x)}{F(\b)}} \right] \\
& \leq  (1-\bar{G}(\b))^{n(1-\epsilon)} \left[ \eta_{x_\epsilon}^{n\epsilon} \vphantom{ \int_{0}^\b} +  \int_{0}^\b{\left(1-\bar{F}(x)^{\frac{1}{\alpha(1-\epsilon)}}\right)^{n \epsilon}\frac{dF(x)}{F(\b)}} \right] ,  \end{align*}
where $\eta_{x_\epsilon} = 1-\bar{G}(x_\epsilon),$ and the last inequality follows from \eqref{eq:prop1}; in case $x_\epsilon \geq \b \geq \b_0$, the integral in the second inequality is zero and the last inequality trivially holds. Now, by extending the preceding integral to $\infty$, we obtain
\begin{align*}
\Pr[N_\b > n] & \leq  \frac{1}{F(\b)}(1-\bar{G}(\b))^{n(1-\epsilon)} \left[ \eta_{x_\epsilon}^{n\epsilon}F(\b) \vphantom{ \int_{0}^\b} + \int_{0}^{\infty}{\left(1-\bar{F}(x)^{\frac{1}{\alpha(1-\epsilon)}} \right)^{n \epsilon}dF(x)} \right] \\
& =  \frac{1}{F(\b)}(1-\bar{G}(\b))^{n(1-\epsilon)}  \left[ \eta_{x_\epsilon}^{n\epsilon} F(\b) \vphantom{ \int_{0}^\b}  +   \E \left (1-\bar{F}(L)^{\frac{1}{\alpha(1-\epsilon)}} \right) ^{n \epsilon}  \right] \\
& \leq  \frac{1}{F(\b)}(1-\bar{G}(\b))^{n(1-\epsilon)}  \left[  \eta_{x_\epsilon}^{n\epsilon} F(\b) \vphantom{ \int_{0}^\b}  +  \E  \; e^{-\bar{F}(L)^{\frac{1}{\alpha(1-\epsilon)}}n \epsilon} \right],\\
\intertext{where we use the elementary inequality $1-x \leq e^{-x}$, $x\geq 0$, and thus}
\Pr[N_\b > n] & \leq  \frac{1}{F(\b)}(1-\bar{G}(\b))^{n(1-\epsilon)}  \left[  \eta_{x_\epsilon}^{n\epsilon} F(\b) +\E  \; e^{-U^{\frac{1}{\alpha(1-\epsilon)}}n \epsilon} \right], 
\intertext{by $\bar{F}(L) = U$, where $U$ is uniformly distributed on $[0,1]$ by Proposition 2.1 in Chapter 10 of \cite{Ross02}; see also the Appendix. Hence,} 
\Pr[N_\b > n] & \leq \frac{1}{F(\b)}(1-\bar{G}(\b))^{n(1-\epsilon)}  \left[  \eta_{x_\epsilon}^{n\epsilon} F(\b) \vphantom{ \int_{0}^\b}  +  \int_{0}^{1}  {e^{-x^{\frac{1}{\alpha(1-\epsilon)}}n \epsilon} dx} \right]\\
& =  \frac{1}{F(\b)}(1-\bar{G}(\b))^{n(1-\epsilon)}  \left[  \eta_{x_\epsilon}^{n\epsilon} F(\b)  \vphantom{ \int_{0}^\b}  +\int_{0}^{n\epsilon}  {\frac{\alpha(1-\epsilon)}{(n\epsilon)^{\alpha(1-\epsilon)}} e^{-z} z^{\alpha(1-\epsilon)-1} dz} \right]\\
& \leq \frac{1}{F(\b)}(1-\bar{G}(\b))^{n(1-\epsilon)}  \left[  \eta_{x_\epsilon}^{n\epsilon} F(\b)  \vphantom{ \int_{0}^\b}  +  \frac{\alpha(1-\epsilon)} {(n\epsilon)^{\alpha(1-\epsilon)}}  \int_0^{\infty} { e^{-z} z^{\alpha(1-\epsilon)-1} dz} \right]\\
& =  \frac{1}{F(\b)}(1-\bar{G}(\b))^{n(1-\epsilon)}  \left[  \eta_{x_\epsilon}^{n\epsilon} F(\b)  \vphantom{ \int_{0}^\b}+ \frac{\alpha(1-\epsilon)} {(n\epsilon)^{\alpha(1-\epsilon)}}   \Gamma(\alpha(1-\epsilon)) \right], \\
\intertext{which follows from the definition of the Gamma function $\Gamma(a)=\int_{0}^{\infty} e^{-t} t^{a-1} dt$. Therefore,}
\Pr[N_\b > n] & \leq  (1-\bar{G}(\b))^{n(1-\epsilon)}  \left[  \eta_{x_\epsilon}^{n\epsilon} +  \frac{\alpha(1-\epsilon)} {F(\b) (n\epsilon)^{\alpha(1-\epsilon)}}   \Gamma(\alpha(1-\epsilon)) \right]\\
& \leq  (1-\bar{G}(\b))^{n(1-\epsilon)}  \left[  \eta_{x_\epsilon}^{n\epsilon} +  \frac{\alpha(1-\epsilon) \epsilon^{-\alpha(1-\epsilon)}} {F(\b_0) n^{\alpha(1-\epsilon)}}   \Gamma(\alpha(1-\epsilon)) \right]\\
& = (1-\bar{G}(\b))^{n(1-\epsilon)}  \left[  \eta_{x_\epsilon}^{n\epsilon} +  \frac{H_{\epsilon}} { n^{\alpha(1-\epsilon)}}  \right],
\end{align*}
since $\b \geq \b_0$, whereas, in the last inequality, we set $H_{\epsilon} =\alpha(1-\epsilon) \epsilon^{-\alpha(1-\epsilon)} \Gamma(\alpha(1-\epsilon))  / F(\b_0)$.

\noindent Now, we can choose $n_0$, such that for any $\epsilon > 0$ and for all $n \geq n_0 $, $\eta_{x_\epsilon}^{n\epsilon}  \leq \epsilon H_{\epsilon} n^{-\alpha(1-\epsilon)}$, so that
\begin{align*}
\Pr[N_\b>n] & \leq    (1-\bar{G}(\b))^{n(1-\epsilon)}  \left[  \epsilon   \frac{H_{\epsilon}} { n^{\alpha(1-\epsilon)}} +  \frac{H_{\epsilon}} { n^{\alpha(1-\epsilon)}}  \right] \\
& = (1-\bar{G}(\b))^{n(1-\epsilon)}   \frac{H_{\epsilon}} { n^{\alpha(1-\epsilon)}}  (1+\epsilon), 
\end{align*}
and by taking the logarithm in the preceding expression, we obtain
\begin{align*}
 \log \Pr[N_\b>n]  & \leq   \log {(H_{\epsilon} (1+\epsilon)) } + n(1-\epsilon) \log (1-\bar{G}(b))  - \alpha(1-\epsilon) \log n  \\
& = \log {(H_{\epsilon} (1+\epsilon) )} + (1-\epsilon) \left[ n \log (1-\bar{G}(b)) - \alpha \log n \right].
\end{align*}

Next, since $-n\log(1-\bar{G}(\b))>0$ and $\alpha \log n > 0$, $n > 1$,
\begin{align*}
\frac{\log \Pr[N_\b>n]}{-n\log(1-\bar{G}(\b)) + \alpha \log n} & \leq  \frac{\log {(H_{\epsilon} (1+\epsilon)) }}{-n\log(1-\bar{G}(\b)) + \alpha \log n} - (1-\epsilon) \\
 & \leq  \frac{\log {(H_{\epsilon} (1+\epsilon)) }}{ \alpha \log n} - (1-\epsilon) 
\end{align*}
and $\alpha \log n$ being increasing in $n$, we can choose $n_0$ such that for any $n  \geq n_0$, $$\frac{\log {(H_{\epsilon} (1+\epsilon)) }}{\alpha \log n} \leq \epsilon .$$
Thus,
\begin{equation}
\frac{\log \Pr[N_\b>n]}{-n\log(1-\bar{G}(\b)) + \alpha \log n}  \leq  - (1-2 \epsilon), 
\end{equation}
which completes the proof by replacing $\epsilon$ with $\epsilon /2.$\qed \end{proof}

Next, we determine the region where the power law asymptotics holds on the logarithmic scale.

\begin{theorem} \label{thm:1}
If
\begin{equation} \label{eq:clog} 
\log \Pr[L > x] \sim  \alpha \log \Pr[A > x] \qquad \text{as } x \rightarrow \infty, 
\end{equation}
$\alpha >0$, then, for any $\epsilon > 0$, there exists positive $n_0$, such that for all $n \geq n_0$, for which $n^{ 1+ \epsilon} \Pr[A>\b] \leq 1$, we have
\begin{equation} \label{eq:thm1-a}
\left| \frac{- \log \Pr[N_{\b} > n]}{ \alpha \log n} - 1 \right| \leq \epsilon.
\end{equation}
\end{theorem}

Note that this result appeared in Theorem 1 of  \cite{TS2010}; for reasons of completeness, we present a simple \textbf{proof} in Section~\ref{sec:4}. 

This result holds in the region $n^{1+\epsilon} \leq O(1/\bar{G}(\b))$. Also, note that one can easily characterize the logarithmic asymptotics of the very end of the exponential tail of $\Pr[N_{\b} >n]$ for small $\b$ and large $n$. In particular, for fixed $\b$, it can be shown that $\log \Pr[N_{\b} >n] \sim n \log(1-\bar{G}(\b))$ as $n \rightarrow \infty$, see Theorem 1 in \cite{TS2010}. However, our objective in this paper is to determine the entire body of the distribution of $\Pr[N_{\b} >n]$ uniformly in $n$ and $\b$.

Next, we extend Theorem~\ref{thm:1} to the entire region $n \geq n_0, \b \geq \b_0$, which includes the geometric term $\Pr[A \leq \b]^n$. For this theorem, we need slightly more restrictive assumptions. The reason why this is the case is that $\Pr[N_\b > n]$ behaves like a power law in the region where $n = o(\log n/ \bar{G}(\b))$, while for $n >> \log n/ \bar{G}(\b)$, it follows essentially a geometric distribution; see Theorem~\ref{thm:2} below. Hence, more restrictive assumptions are required since the geometric distribution is much more sensitive to the changes in its parameters (informally, $((1+\epsilon)x)^{-\alpha} \approx x^{-\alpha}$ but $e^{-(1+\epsilon)x} \not\approx e^{-x}$).

\begin{definition}
A function $\ell(x)$ is slowly varying if $  \ell(x) / \ell(\lambda x)  \rightarrow 1$ as $x \rightarrow \infty$ for any fixed $\lambda >0$. 
\end{definition}
If not directly implied by our assumptions, $\ell(x)$ is assumed positive and locally bounded.

\begin{theorem} \label{thm:2}
If $\Pr[L > x] = \ell(\Pr[A> x]^{-1}) \Pr[A> x]^{\alpha} $, for $\alpha >0$, $\ell(x)$ slowly varying, then for any $\epsilon >0 $, there exist $n_0, \b_0,$ such that for all $n \ge n_0, \b \ge \b_0$,
\begin{equation}  \label{eq:thm2-a}
\left| \frac{-\log \Pr[N_{\b}>n]}{-n \log \Pr[A \leq \b]+ \alpha \log n} - 1 \right| \leq\epsilon.
\end{equation}
\end{theorem} 
\begin{proof}
See Section~\ref{sec:4}. 
\end{proof}

\begin{remark}
Note that the statement of this theorem can be formulated in an equivalent form 
\begin{equation*} 
\left| \frac{-\log \Pr[N_{\b}>n]}{n  \Pr[A > \b] + \alpha \log n} - 1 \right| \leq\epsilon,
\end{equation*}since $-n \log \Pr[A \leq \b] \sim n  \Pr[A > \b] $ as $\b \rightarrow \infty$.
\end{remark}

\begin{remark}
This theorem extends Theorem~1 in \cite{TS2010}. In particular, it proves the result uniformly in $n$ and $\b$, while Theorem~1 in \cite{TS2010} characterized the initial power law part of the distribution ($n \leq \bar{G}(\b)^{-\eta}, 0<\eta <1$) and the very end with exponential tail (fixed $\b$, $n \rightarrow \infty$). 
\end{remark}


\subsection{Exact Asymptotics}\label{sec:exact}
In this section, we derive the exact approximation for $\Pr[N_b >n]$ that works uniformly for all $n, \b$ sufficiently large (Theorems~\ref{thm:3} and \ref{thm:4}). As noted earlier in the introduction, this characterization is explicit in that it is a product of a power law and the Gamma distribution
\begin{equation} \label{eq:approx}
\Pr[N_\b >n] \approx \frac{\alpha} {n^\alpha \ell(n \wedge \Pr[A>\b]^{-1})} \int_{-n \log \Pr[A \leq \b]}^{\infty} e^{-z} z^{\alpha-1} dz ,
\end{equation}
where $x \wedge y = \min (x,y)$ and $\ell(\cdot)$ is slowly varying. Implicitly, the argument of $\ell(x)$ is altered depending on whether $n\Pr[A>\b]  \leq C$ or $n \Pr[A>\b] > C$ for some constant $C$. Hence, we can choose $C=1$ since $\ell(n \wedge 1/ \Pr[A>\b]) \approx \ell(n \wedge C /\Pr[A>\b])$ for large $n, \b$.
Note that when $-n\log \Pr[A \leq \b] \downarrow 0$, the power law dominates, whereas when $-n\log \Pr[A \leq \b] \rightarrow \infty$, the integral determines the tail with the geometric (exponential) leading term.

We would like to point out that approximation \eqref{eq:approx} actually works well when $\Pr[A>\b]^{-1}$ is large rather than $\b$; this can be concluded by examining the proofs of the theorems in this section. Hence, formula \eqref{eq:approx} can be accurate for relatively small values of $\b$ provided that $A$ is light-tailed. This may be the reason why we obtain accurate results in our simulation examples in Section~\ref{s:sim} for small values of $\b$.

First, in Theorem \ref{thm:3}, we precisely describe the region where the distribution of $N_\b$ exhibits the power law behavior, $n\Pr[A>\b] \leq C$, for any fixed constant $C$. Then, Theorem~\ref{thm:4} covers the remaining region, $n\Pr[A > \b] > C$, where $\Pr[N_\b > n]$ approaches the geometric tail. Additional discussion of the results and the treatment of some special cases are presented at the end of this section; see Propositions~\ref{prop:integer} and \ref{prop:exptail}.

\begin{theorem}\label{thm:3}
Let $\Pr[L>x]^{-1} = \ell (\Pr[A>x]^{-1})\Pr[A>x]^{-\alpha}$, $\alpha >0$, $x \geq 0$, and $C>0$ be a fixed constant. Then, for any $\epsilon >0$, there exists $n_0$ such that for all $n > n_0$, and $n \Pr[A > \b] \leq C$,
\begin{equation}\label{eq:thm21}
\left| \frac{\Pr[N_{\b} >n] n^\alpha \ell(n) } { \alpha \Gamma(-n \log \Pr[A \leq \b], \alpha)} -1  \right| \leq \epsilon,
\end{equation}
where $\Gamma(x, \alpha)$ is the incomplete Gamma function defined as $\int_{x}^{\infty} e^{-z} z^{\alpha-1} dz$.
\end{theorem}

\begin{remark}
Related result was derived in Theorem 3 of \cite{TS2010} where it was required that $n \leq \bar{G}(\b)^{-\eta}, 0< \eta < 1$. Note that here we broaden the region where the result holds by requiring $n \leq C / \bar{G}(\b)$, which is larger than $n \leq \bar{G}(\b)^{-\eta}$. Furthermore, this is the largest region where the exact power law asymptotics $O(n^{-\alpha}/\ell(n))$ holds since for $n\bar{G}(\b) > C,  \Gamma(n\bar{G}(\b), \alpha) \leq \Gamma(C, \alpha) \rightarrow 0$ as $C \rightarrow \infty$. 
\end{remark}

\begin{remark} \label{re:2}
Note here that the incomplete Gamma function $\Gamma(\alpha, x) = \int_{x}^{\infty}{z^{\alpha - 1}e^{-z}dz}$ can be easily computed using the well known asymptotic approximation (see Sections 6.5.32 in \cite{abram64}), as $x \rightarrow \infty$, 
\begin{displaymath}
\Gamma(\alpha, x) \sim x^{\alpha-1}e^{-x} \left[ 1+ \frac{\alpha-1}{x}+  \frac{(\alpha-1)(\alpha-2)}{x^2} + \dots \right].
\end{displaymath}
\end{remark}

\begin{proof}This proof uses some of the ideas from the proof of Theorem~2.1 in \cite{JT2007}. However, it is much more involved since one has to incorporate the assumption $n \Pr[A > \b] \leq C$, which ensures the power law body.

Let $\Phi(x) = \ell(x) x^{\alpha}$. Then, $\Phi(x)$ is regularly varying with index $\alpha$ and, thus, for any $c>0$,
\begin{equation*}\label{eq:drv}
\lim_{x\rightarrow \infty} \frac{\Phi(cx)}{\Phi(x)} = c^{\alpha} < \infty,
\end{equation*}
and, in particular, we can choose $c = e$, which implies that there exists $n_\epsilon$ such that for $n/e^k > n_\epsilon$,
\begin{align}\label{eq:drv1}
  \frac{\Phi(n)}{\Phi(n/e^k)} \leq e^{ k(\alpha + 1)}.
\end{align}
Without loss of generality, we may assume that $\Phi(.)$ is eventually absolutely continuous, strictly monotone and locally bounded for $x >0$  since we can always find an absolutely continuous and strictly monotone function 
\begin{equation} \label{eq:fstar}
\Phi^* (x) =\begin{cases} \alpha \int_{1}^x \Phi(s) s^{-1} ds,  & x \geq 1   \\
0, &  0 \leq x < 1, \end{cases} 
\end{equation}
which for $x$ large enough satisfies 
\begin{displaymath}
\bar{F}(x)^{-1} = \Phi (\bar{G}(x)^{-1})  \sim \Phi^* (\bar{G}(x)^{-1}).
\end{displaymath}
This implies that, for any $0 < \epsilon <1$ and $x \geq x_0$, we have
\begin{equation} \label{eq:main1}
1/ \Phin \left((1+\epsilon) \bar F (x)^{-1} \right)  \leq \bar{G}(x) \leq 1/ \Phin \left((1-\epsilon) \bar F (x)^{-1} \right),
\end{equation}
where $\Phin(\cdot)$ denotes the inverse function of $\Phi^*(\cdot)$; note that the monotonicity of $\Phi^*(x) $, for all $x \geq 1$, guarantees that its inverse exists. To simplify the notation in this proof, we shall use $\Phi(\cdot)$ to denote $\Phi^*(\cdot)$. Furthermore, $\Phin(\cdot)$ is regularly varying with index $1/ \alpha$ (see Theorem~1.5.12 in \cite{BG87}), implying that 
\begin{displaymath}
\Phin ( (1+\epsilon) x ) \sim \left( \frac{1+\epsilon}{1-\epsilon} \right)^{1/\alpha} \Phin ( (1-\epsilon) x) ,
\end{displaymath}
as $x \rightarrow \infty$. Therefore, for $\eta_\epsilon = \eta(\epsilon) =  \left[(1+\epsilon)/(1-\epsilon) \right]^{2/\alpha}$ and $x$ large,
\begin{equation} \label{eq:Phin}
\eta_\epsilon^{-1} \bar G (x) \leq 1/ \Phin \left((1+\epsilon) \bar F (x)^{-1} \right) \leq 1/ \Phin \left((1-\epsilon) \bar F (x)^{-1} \right) \leq \eta_\epsilon \bar G (x).
\end{equation} 

First, notice that the number of retransmissions is geometrically distributed given the data size $L_\b$,
\begin{align} \label{eq:thm3-0}
\nonumber \Pr[N_\b > n] & = \E [1-\bar{G}(L_\b)]^n \\
& =  \E [1-\bar{G}(L_\b)]^n \textbf{1}(L_\b \leq x_0)   + \E [1-\bar{G}(L_\b)]^n \textbf{1}(L_\b>x_0).
\end{align}


We begin with the \emph{lower bound}. For $H > C$ and $x_0$ as in \eqref{eq:thm3-0}, we choose $x_n > x_0$ such that $\Phin((1-\epsilon) \bar F(x_n)^{-1}) = n / H$, for $n$ large, and thus,
\begin{align*}
\Pr[N_\b > n] & = \E [1-\bar{G}(L_\b)]^n \\
 &  \geq  \E \left[(1-\bar{G}(L_\b))^n \textbf{1}(L_\b>x_n) \right]  \\
  &  \geq  \E \left[ \left(1- \frac{1}{\Phin((1-\epsilon)\bar{F}(L_\b)^{-1})} \right)^n \textbf{1}(L_\b  > x_n )   \right] \\
  &  =  \int_{x_n}^{\b} \left(1- \frac{1}{\Phin((1-\epsilon)\bar{F}(x)^{-1})} \right)^n  \frac{dF(x)}{F(\b)},  \\
  \intertext{where we use our main assumption \eqref{eq:main1}. Now, since $F(\b) \leq 1$ and using the continuity of $F(x)$ and change of variables $z= n/\Phin((1-\epsilon)\bar F(x)^{-1})$, we obtain, }
  \Pr[N_\b > n] & \geq \int_{n/\Phin((1-\epsilon)\bar{F}(\b)^{-1})}^H  \left(1- \frac{z}{n} \right)^n  \frac{\Phi'(n/z)}{\Phi^2(n/z)} \frac{(1-\epsilon)n}{z^2} dz \\
  & \geq \int_{ \eta_\epsilon n \bar G (\b)}^H  \left(1- \frac{z}{n} \right)^n  \frac{\Phi'(n/z)}{\Phi^2(n/z)} \frac{(1-\epsilon)n}{z^2} dz,
\end{align*}
where we use that $ \eta_\epsilon \bar G (\b) \geq  1/\Phin((1-\epsilon)\bar{F}(\b)^{-1})$ from \eqref{eq:Phin}, which holds for $\b$ large, or equivalently $n$ large by our assumption $n \bar G(\b) \leq C$. Now, we consider two distinct cases:  

If $ \eta_\epsilon n \bar G (\b) < h$, where $h>0$ is a small constant, then
\begin{align*}
\Pr[N_\b > n]  &  \geq  (1-\epsilon)  \int_{h}^{H}  \left(1- \frac{z}{n} \right)^n \frac{\Phi'(n/z)}{\Phi^2(n/z)} \frac{n}{z^2} dz \\
 & \geq  (1-\epsilon)^{3/2}  \frac{\alpha}{\Phi(n)} \int_{h}^{H} \left(1- \frac{z}{n} \right)^nz^{\alpha-1}  dz ,
 \intertext{where we use the properties of regularly varying functions that for all $h \leq z\leq H$ and large $n$,  
\begin{displaymath} 
 \frac{\Phi(n)}{ \Phi(n/z)} \geq (1-\epsilon)^{1/2} z^\alpha,
 \end{displaymath}
  and 
  \begin{displaymath} 
  \Phi'(n/z)/\Phi(n/z) =  \frac{\alpha z}{n}, \end{displaymath}  
  for $n > H $ [see \eqref{eq:fstar}]. Next, using $1-x \geq e^{-(1+\delta)x}$ for  $\delta >0$ and $0 \leq x \leq x_\delta$, for $n$ large enough ($n > H / x_\delta$) we obtain }
  \Pr [N_\b > n] & \geq  (1-\epsilon)^{3/2}  \frac{\alpha}{\Phi(n)} \int_{h}^{H} e^{-(1+\delta)z} z^{\alpha-1}  dz \\
& \geq  (1-\epsilon)^{3/2} e^{-\delta H} \frac{\alpha}{ \Phi(n)}  \int_{h}^{H} e^{-z} z^{\alpha-1}  dz, 
\intertext{and by choosing $\delta > 1/H$ so that $e^{-\delta H} \geq (1-\epsilon)^{1/2}$, we have}
 \Pr [N_\b > n] & \geq  (1-\epsilon)^2 \frac{\alpha}{ \Phi(n)}  \int_{h}^{H} e^{-z} z^{\alpha-1}  dz  \\
 & \geq  (1-\epsilon)^{2}  \frac{\alpha}{ \Phi(n)} \left[ \int_{n \bar G (\b)}^{H} e^{-z} z^{\alpha-1}  dz -  \int_{n \bar G (\b)}^{h} e^{-z} z^{\alpha-1}  dz  \right] \\
   & \geq  (1-\epsilon)^2  \frac{\alpha}{ \Phi(n)} \left[ \int_{n \bar G (\b)}^{H} e^{-z} z^{\alpha-1}  dz -  \int_{0}^{h} e^{-z} z^{\alpha-1}  dz  \right] 
   \end{align*}
\begin{align*}
    & \geq  (1-\epsilon)^{2}  \frac{\alpha}{\Phi(n)} \left[ \int_{n \bar G (\b)}^{\infty} e^{-z} z^{\alpha-1}  dz  - \int_{H}^{\infty} e^{-z} z^{\alpha-1}  dz -   \frac{h^{\alpha }}{\alpha}  \right] \\
       & \geq  (1-\epsilon)^{2}  \frac{\alpha}{\Phi(n)} \left[ \int_{n \bar G (\b)}^{\infty} e^{-z} z^{\alpha-1}  dz  - 2 e^{-H}H^{\alpha -1} -   \frac{h^{\alpha }}{\alpha}  \right] \\
       & =  (1-\epsilon)^{2}  \frac{\alpha}{\Phi(n)} \Gamma(n \bar G (\b), \alpha)  \left( 1  - \frac {2 e^{-H}H^{\alpha -1} +   h^{\alpha} /\alpha }{\Gamma(C, \alpha)} \right),
\end{align*}
where the second to last inequality follows from the approximation for the incomplete gamma function for large $H$ [see Remark~\ref{re:2} of Theorem~\ref{thm:3}] and the last inequality uses the assumption $n \bar G (\b) \leq C$. Now, picking $H, h$ such that $  2 e^{-H}H^{\alpha -1}  +  h^{\alpha }/\alpha \leq \epsilon  \Gamma(C, \alpha) $, yields 
\begin{equation*}
\Pr[N_\b > n]  \geq (1-3\epsilon) \frac{\alpha }{F(\b) \Phi(n)} \Gamma(n \bar{G}(\b), \alpha).
\end{equation*}

If $h \leq \eta_\epsilon n \bar G (\b) \leq C$, then
\begin{align*}
\Pr[N_\b > n]  & \geq  (1-\epsilon) \int_{ \eta_\epsilon n \bar G(\b)}^{H}e^{-z}   \frac{\Phi'(n/z)}{\Phi^2(n/z)} \frac{n}{z^2} dz \\
    & \geq  (1-\epsilon)^2 \frac{ \alpha}{ \Phi(n)} \int_{ \eta_\epsilon n \bar G(\b)}^{H}e^{-z}   z^{\alpha-1} dz,
    \end{align*}
which follows from the regularly varying properties in the region $h/  \eta_\epsilon < n \bar G(\b) \leq z \leq H$. 

For the preceding integral, similarly as before, we have
\begin{align*} 
\int_{ \eta_\epsilon n \bar G(\b)}^{H}e^{-z}   z^{\alpha-1} dz & =  \int_{ n \bar G(\b)}^{H}e^{-z}   z^{\alpha-1} dz  -  \int_{n \bar G(\b)}^{ \eta_\epsilon n \bar G(\b)}e^{-z}   z^{\alpha-1} dz \\
& \geq \int_{ n \bar G(\b)}^{H}e^{-z}   z^{\alpha-1} dz  -  \int_{n \bar G(\b)}^{\eta_\epsilon n \bar G(\b)}  z^{\alpha-1} dz \\
& = \int_{ n \bar G(\b)}^{H}e^{-z}   z^{\alpha-1} dz  -    (n \bar G(\b))^\alpha \frac{ \eta_\epsilon^\alpha-1}{\alpha} \\
& \geq \int_{ n \bar G(\b)}^{\infty}e^{-z}   z^{\alpha-1} dz  - \int_{H}^{\infty}e^{-z}   z^{\alpha-1} dz -    
C^\alpha \frac{ \eta_\epsilon^\alpha-1}{\alpha}  \\
& \geq \int_{ n \bar G(\b)}^{\infty}e^{-z}   z^{\alpha-1} dz  -  2 e^{-H}H^{\alpha -1} -  \frac{4 \epsilon C ^\alpha}{\alpha}  \\
& = \Gamma( n \bar G(\b),\alpha) \left(1  -  \frac{2 e^{-H}H^{\alpha -1} + 4 \epsilon C^\alpha / \alpha }{ \Gamma( C,\alpha)} \right),
\end{align*}
where we use the approximation for the incomplete gamma function for large $H$, that $\eta_\epsilon^\alpha - 1 \rightarrow 4 \epsilon$ as $\epsilon \rightarrow 0$ and $n \bar G(\b) \leq C$. Now, letting $H $ be such that $ 2 e^{-H}H^{\alpha -1} + 4 \epsilon C^\alpha / \alpha  \leq \sqrt \epsilon \Gamma(C, \alpha)$ yields 
\begin{equation*}
\Pr[N_\b > n]   \geq (1-\epsilon)^2 (1 - \sqrt \epsilon )  \frac{\alpha }{F(\b) \Phi(n)} \Gamma(n \bar{G}(\b), \alpha).
\end{equation*}

Finally, since $\bar G(\b) \leq - \log (1- \bar G (\b))$, we obtain 
\begin{equation} \label{eq:I1-lb}
\Pr[N_\b > n]  \geq (1-\epsilon)^2 (1 - \sqrt \epsilon ) \frac{\alpha }{F(\b) \Phi(n)} \Gamma(- n\log (1- \bar G (\b)), \alpha),
\end{equation}
which proves the lower bound after replacing $(1-\epsilon)^2 (1 - \sqrt \epsilon )$ with $1- \epsilon$.

Next, we derive the \emph{upper bound}. Note that for $x_0$ as in \eqref{eq:main1},
\begin{align} \label{eq:0all}
 \nonumber \Pr[N_\b > n]&= \E [1-\bar{G}(L_\b)]^n \\
 \nonumber &  \leq  (1-\bar G(x_0))^n +  \E [1-\bar{G}(L_\b)]^n \textbf{1}(L_\b>x_0)    \\
   &  \leq  e^{- n \bar G(x_0)} +  \E \left(1-\frac{1}{\Phin((1+\epsilon)\bar{F}(L_\b)^{-1})} \right)^n \textbf{1}(L_\b>x_0)  ,
\end{align}
which follows from \eqref{eq:main1} and the elementary inequality $1-x \leq e^{-x}$. Now, for any $H >  \max{(C, 1)}$, we obtain 
\begin{align}  \label{eq:1all}
\nonumber  & \Pr[N_\b > n]  \leq  e^{- n \bar G(x_0)}  +  \E \left(1-\frac{1}{\Phin((1+\epsilon)\bar{F}(L_\b)^{-1})} \right)^n \textbf{1} \left( L_\b > x_0 \right)   \\
\nonumber  & \leq  e^{- n \bar G(x_0)} +  \ \E \left(1-\frac{1}{\Phin((1+\epsilon)\bar{F}(L_\b)^{-1})} \right)^n \textbf{1} \left( \frac{1}{\Phin((1+\epsilon)\bar{F}(L_\b)^{-1})} < \frac{H}{n} \right)   \\
 \nonumber & + \sum_{k= \lfloor \log H \rfloor}^{\lceil \log (n/n_\epsilon) \rceil} e^{- e^k} \Pr \left[ e^k \leq \frac{n}{\Phin((1+\epsilon)\bar{F}(L_\b)^{-1})} \leq  e^{k+1} \right] + e^{-n/n_\epsilon} \\ 
   & \eqdef I_0 +  I_1 + I_2 + I_3. 
 \end{align}

First, we upper bound $I_1$ in \eqref{eq:1all}, which equals
 \begin{align*}
  I_1 & =   \frac{1}{F(\b)} \int_{0}^{\b} \left(1-\frac{1}{\Phin((1+\epsilon)\bar{F}(x)^{-1})}\right)^n  \textbf{1} \left( \frac{n}{\Phin((1+\epsilon)\bar{F}(x)^{-1})} <H \right)  dF(x) \\
    & =   \frac{1+\epsilon}{\Phi(n)F(\b)}  \int_{n/ \Phin((1+\epsilon)\bar{F}(\b)^{-1})}^{H} \left(1- \frac{z}{n}\right)^n{\frac{\Phi(n)}{\Phi(n/z)} \frac{\Phi'(n/z)}{\Phi(n/z)} \frac{n}{z^2}dz} \\
     & \leq   \frac{1+\epsilon}{\Phi(n) F(\b)}  \int_{n \bar{G}(\b)/ \eta_\epsilon}^{H} \left(1- \frac{z}{n}\right)^n{\frac{\Phi(n)}{\Phi(n/z)} \frac{\Phi'(n/z)}{\Phi(n/z)} \frac{n}{z^2}dz},
 \end{align*}
  where we use the change of variables $z= n/\Phin((1+\epsilon)F(x)^{-1})$ and the absolute continuity of $F(x)$. For the last inequality, observe that $1/ \Phin((1+\epsilon) \bar F (\b)^{-1}) \geq  \bar G(\b) / \eta_\epsilon$ from \eqref{eq:Phin}. Now, similarly as before, we consider two cases: 

If $n \bar G (\b) < \eta_\epsilon h_\epsilon \leq C$, where $h_\epsilon > 0$ is a small constant, $I_1$ is upper bounded by
 \begin{align}\label{eq:I1-a}
  I_1 &  \leq \frac{1+\epsilon}{F(\b) \Phi(n)} \int_{h_\epsilon}^{H} \left(1- \frac{z}{n}\right)^n{\frac{\Phi(n)}{\Phi(n/z)} \frac{\Phi'(n/z)}{\Phi(n/z)} \frac{n}{z^2}dz}   +  \Pr \left(  \frac{n}{ \Phin((1+\epsilon) \bar F (L_\b)^{-1})} < h_\epsilon \right) .
 \end{align}
 
Now, since $\Phi(.)$ is absolutely continuous and regularly varying, it follows that for all $h_\epsilon \leq z \leq H$, 
\begin{displaymath} 
\frac{\Phi(n)}{  \Phi(n/z)}   \leq (1+\epsilon)^{1/2} z^ \alpha ,
\end{displaymath}
for large $n$, and, by \eqref{eq:fstar}, 
\begin{displaymath}  
\frac{\Phi'(n/z)}{  \Phi(n/z)} = \frac{\alpha z }{ n},
\end{displaymath}
for $n > H $. 

Next, we compute the second term in \eqref{eq:I1-a} as
\begin{align*}
 \Pr \left(  \bar F (L_\b) < \frac{1+\epsilon} { \Phi(n/h_\epsilon)} \right) & \leq \int_0^{\infty} \textbf{1} \left(\bar F(x) < \frac{1+\epsilon}{\Phi(n/h_\epsilon )} \right) \frac{dF(x)}{F(\b)}= \frac{1}{F(\b)}\Pr \left[ \bar F(L) < \frac{1+\epsilon}{\Phi(n/h_\epsilon )}\right] \\
 & \leq \frac{1+\epsilon}{F(\b) \Phi(n/h_\epsilon )} \leq \frac{(1+\epsilon)^2 h_\epsilon^\alpha}{\Phi(n)},
\end{align*}
which follows from the uniform distribution of $\bar F(L)$ and using $\Phi(n) / \Phi(n/h_\epsilon)  \leq (1+\epsilon)^{1/2} h_\epsilon^\alpha$ for large $n$, along with $F(\b)^{-1}\leq (1+ \epsilon)^{1/2}$. Now, observe that the first term in \eqref{eq:I1-a} is upper bounded by
\begin{align*}
  \frac{\alpha (1+\epsilon)^2}{\Phi(n)}  \int_{h_\epsilon}^{H} \left(1- \frac{z}{n}\right)^n z^{\alpha-1} dz    \leq   \frac{\alpha (1+\epsilon)^2}{\Phi(n)}   \int_{n \bar{G}(\b)/ \eta_\epsilon}^{H}  \left(1- \frac{z}{n}\right)^n z^{\alpha-1} dz,
\end{align*}
since $ n \bar{G}(\b) < h_\epsilon \eta_\epsilon$. Also, for the integral we obtain 
\begin{align*}
 \int_{n \bar{G}(\b)/ \eta_\epsilon}^{H}  \left(1- \frac{z}{n}\right)^n z^{\alpha-1} dz & =  \int_{n \bar{G}(\b)}^{H}  \left(1- \frac{z}{n}\right)^n z^{\alpha-1} dz + \int_{n \bar{G}(\b)/ \eta_\epsilon}^{n \bar{G}(\b)}  \left(1- \frac{z}{n}\right)^n z^{\alpha-1} dz \\
 & \leq \int_{n \bar{G}(\b)}^{H}  \left(1- \frac{z}{n}\right)^n z^{\alpha-1} dz + \int_{n \bar{G}(\b)/ \eta_\epsilon}^{n \bar{G}(\b)}  z^{\alpha-1} dz \\
 & \leq \int_{n \bar{G}(\b)}^{H}  \left(1- \frac{z}{n}\right)^n z^{\alpha-1} dz + (n \bar{G}(\b))^ {\alpha} (1 -  \eta_\epsilon^{-\alpha} )/ \alpha \\
  & \leq \int_{n \bar{G}(\b)}^{H}  \left(1- \frac{z}{n}\right)^n z^{\alpha-1} dz + 5 C^ {\alpha}  \epsilon / \alpha,
\end{align*}
after observing that $1 -  \eta_\epsilon^{-\alpha} \rightarrow 4 \epsilon$ as $\epsilon \rightarrow 0$. Now, by changing the variables $ 1- z/n = e^{-u/n}$, we have
\begin{align}  \label{eq:I1upperb}
\nonumber I_1 &  \leq      \frac{\alpha (1+\epsilon)^2}{\Phi(n)}  \int_{n \bar G(\b)}^{H} \left(1- \frac{z}{n}\right)^n z^{\alpha-1} dz  + \frac{(1+\epsilon)^2 5 C^\alpha \epsilon}{\Phi(n)} + \frac{(1+\epsilon)^2 h_\epsilon^\alpha}{\Phi(n)}  \\
\nonumber & \leq   \frac{\alpha (1+\epsilon)^2}{\Phi(n)} \int_{- n \log(1 - \bar{G}(\b))}^{- n \log (1 - H/n)}  e^{-u} (1 - e^{-u/n})^{\alpha-1} n^{\alpha -1} e^{- u/n } du + \frac{(1+\epsilon)^2 (5 C^\alpha \epsilon + h_\epsilon^\alpha)}{\Phi(n)} \\
 \nonumber & \leq   \frac{\alpha (1+\epsilon)^2}{\Phi(n)} \int_{- n \log(1 - \bar{G}(\b))}^{\infty}  e^{-u} u^{\alpha-1} du + \frac{(1+\epsilon)^2}{\Phi(n)}  (5 C^\alpha \epsilon+h_\epsilon^\alpha)  \\
 & \leq   \frac{\alpha (1+\epsilon)^2}{\Phi(n)} \int_{- n \log(1 - \bar{G}(\b))}^{\infty}  e^{-u} u^{\alpha-1} du \left[  1 + \frac{5 C^\alpha \epsilon+ h_\epsilon^\alpha}{\alpha \Gamma (2C, \alpha)} \right] \\
\nonumber & \leq   \frac{\alpha (1+\epsilon)^2 (1+\sqrt{\epsilon})}{\Phi(n)} \int_{- n \log(1 - \bar{G}(\b))}^{\infty}  e^{-u} u^{\alpha-1} du,
\end{align}
where, in the second inequality, we use $e^{-u/n} \leq 1$, the inequality $1- e^{-x} \leq x$, $x \geq 0$ and extend the integral to infinity. Last, we pick $\epsilon$ small, such that $ 5 C^\alpha \epsilon+h_\epsilon^{\alpha} \leq  \sqrt {\epsilon} \alpha   \Gamma(2C, \alpha)$. 

Note that the preceding equation along with \eqref{eq:I1-lb} imply that $I_1$ is lower bounded as $I_1 \geq  (1-\epsilon) \alpha \Gamma(2n \bar G(\b), \alpha) / \Phi(n) \geq  ( \alpha /2) \Gamma(2C, \alpha) / \Phi(n)$, for all $n > n_0$ and $\epsilon < 1/2$, by the inequality $1-x \geq e^{-2x}$ for $x\geq 0$ small, since by assumption $n \bar G(\b) \leq C$, i.e., $\bar G(\b)$ is small. 

If $h_\epsilon \eta_\epsilon  \leq n \bar G (\b) \leq C$, we have 
\begin{align*}
I_1 &  \leq \frac{1+\epsilon}{F(\b) \Phi(n)} \int_{n\bar{G}(\b)/  \eta_\epsilon}^{H} \left(1-\frac{z}{n} \right)^n {\frac{\Phi(n)}{\Phi(n/z)} \frac{\Phi'(n/z)}{\Phi(n/z)} \frac{n}{z^2}dz},
\end{align*}
and, by the properties of regularly varying functions in the interval $n/H  \leq n/z \leq 1/ \bar{G}(\b) \leq n/h_\epsilon$, for $H > C$, and using the same arguments as in \eqref{eq:I1upperb}, we have
\begin{align*}
 I_1 & \leq   \frac{\alpha (1+\epsilon)^2}{\Phi(n)}  \int_{n \bar G (\b)/  \eta_\epsilon}^{H} \left(1-\frac{z}{n} \right)^n z^{\alpha-1} dz \\
 & \leq   \frac{\alpha (1+\epsilon)^2}{\Phi(n)}  \int_{- n \log(1 -   \bar{G}(\b))}^{\infty} e^{-z} z^{\alpha-1} dz \left[1 + \frac{5 C^ {\alpha}  \epsilon }{\alpha \Gamma(2C, \alpha) }  \right]  \\
 & \leq   \frac{\alpha  (1+\epsilon)^2 (1+\sqrt{\epsilon})}{\Phi(n)}  \int_{- n \log(1 - \bar{G}(\b))}^{\infty} e^{-z} z^{\alpha-1} dz.
 \end{align*}

Therefore, from both cases, it follows that for all $n > n_0$, 

\begin{equation} \label{eq:I1-ub}
I_1  \leq \frac{\alpha(1+\epsilon)}{\Phi(n)} \Gamma(-n \log(1- \bar{G}(\b)), \alpha),
\end{equation}
after replacing $(1 + \epsilon)^{2}(1+ \sqrt{\epsilon})$ with $1 + \epsilon$.

Next, we evaluate the second term in \eqref{eq:1all} as
\begin{align*}
 I_2 & = \sum_{k=\lfloor \log H \rfloor}^{\lceil \log(n/n_\epsilon) \rceil} e^{-
e^k} \Pr \left[ e^k \leq  \frac{n}{\Phin((1+\epsilon)\bar{F}(L_\b)^{-1})} \leq  e^{k+1} \right] \\
 & = \sum_{k=\lfloor \log H \rfloor}^{\lceil \log(n/n_\epsilon) \rceil} e^{-e^k} \Pr \left[ (1+\epsilon)  / \Phi \left(\frac{n}{ e^{k+1}} \right) \leq \bar{F}(L_\b) \leq (1+\epsilon) / \Phi \left( \frac{n}{ e^k} \right)  \right] 
 \end{align*}
 \begin{align*}
  & \leq \sum_{k=\lfloor \log H \rfloor}^{\lceil \log(n/n_\epsilon) \rceil} e^{-e^k} \int_0^{\infty} \textbf{1}  \left( \bar{F}(x) \leq \frac{1+\epsilon} {\Phi \left(n/ e^k \right)}  \right) \frac{dF(x)}{F(\b)} \\
 &\leq \sum_{k=\lfloor \log H \rfloor}^{\infty} e^{-e^k} \frac{1+\epsilon}{F(\b) \Phi \left(n/ e^{k} \right)}, 
 \end{align*}
which follows from the fact that the integral in the second inequality is equal to $\Pr[\bar F(L) \leq (1+\epsilon) / \Phi \left(n/ e^k \right) ]/ F(\b)$ and $\bar F(L) $ is uniform in $[0,1]$. Thus, 
\begin{align*}
  I_2 &\leq \frac{1+\epsilon}{F(\b)\Phi(n)}\sum_{k=\lfloor \log H \rfloor}^{\infty} e^{-e^k} \frac{\Phi(n)}{\Phi(n/e^{k})} \\
 &\leq \frac{ 1+\epsilon }{F(\b) \Phi(n)}\sum_{k=\lfloor \log H \rfloor}^{\infty} e^{-e^k}e^{ k (\alpha+1)} , 
\end{align*}
where we make use of \eqref{eq:drv1}. Since the preceding sum is finite, we obtain that for large $H$ and all $n > n_0$,
\begin{equation} \label{eq:I2}
I_2 \leq   \frac{\epsilon}{2} I_1.
\end{equation}

Last, we observe that, for fixed $x_0$, it follows that for $n > n_0$, 
 \begin{align} \label{eq:I3}
  I_0 + I_3 = e^{-n \bar G(x_0)} + e^{-n/n_\epsilon} \leq  \frac{\epsilon}{2} I_1.
 \end{align}

Finally, using \eqref{eq:I1-ub}-\eqref{eq:I3}, we obtain for \eqref{eq:1all} that for all $n > n_0$,
\begin{align*}
\Pr[N_\b > n ] \leq (1+\epsilon)^2 \frac{\alpha }{\Phi(n)} \Gamma(-n \log (1-\bar{G}(\b)), \alpha),
\end{align*}
which completes the proof after replacing $ (1 +\epsilon)  $ with $(1+\epsilon )^{1/2}$. 
\qed \end{proof}

The following corollary is an immediate consequence of Theorem~\ref{thm:3} and it represents a small generalization of Theorem 2.1 in \cite{JT2007}.

\begin{corollary}\label{cor:PL}
If $\Pr[L>x]^{-1} = \ell (\Pr[A>x]^{-1})\Pr[A>x]^{-\alpha}$, $x \geq 0, \alpha>0$, where $\ell (x)$ is slowly varying, then, as $n \rightarrow \infty$ and $n \Pr[A>\b] \rightarrow 0 $,
\begin{equation}\label{eq:PL}
\Pr[N_\b > n] \sim \frac{\Gamma(\alpha + 1)}{\ell(n) n^\alpha}.
\end{equation}
\end{corollary}

Now, we characterize the remaining region where $ n \Pr[A> \b] > C$. Informally speaking, this is the region where $\Pr[N_\b > n]$ has a lighter tail converging to the exponential when $n >> \bar{G}(\b)^{-1}$. In the following theorem, we need more restrictive assumptions for $\ell(x)$; see the discussion before Theorem~\ref{thm:2}. In particular, we assume that $\ell(x)$ is slowly varying and eventually differentiable with $\ell'(x) x/ \ell(x) \rightarrow 0$ as $x \rightarrow \infty$. 

\begin{theorem}\label{thm:4}
Assume that $\Pr[L>x]^{-1} = \ell (\Pr[A>x]^{-1})\Pr[A>x]^{-\alpha}, \alpha >0$, $x \geq 0$, where $\ell(x)$ is  slowly varying and eventually differentiable with $\ell'(x) x/ \ell(x) \rightarrow 0$ as $x \rightarrow \infty$. Then, for any $\epsilon > 0$, there exist $\b_0, n_0$, such that for all $n > n_0, \b > \b_0, n\Pr[A>\b] >C$,
\begin{equation} \label{eq:exact}
\left| \frac{\Pr[N_{\b} >n]n^\alpha \ell(\Pr[A>\b]^{-1})  } { \alpha \Gamma(-n \log \Pr[A \leq \b], \alpha)} -1  \right| \leq \epsilon .
\end{equation}
\end{theorem}

\begin{remark}
Observe that Theorems~\ref{thm:3} and \ref{thm:4} cover the entire distribution $\Pr[N_\b > n ]$ for all large $n$ and $\b$. Interestingly, the formula for the approximation is the same except for the argument of the slowly varying part, which equals to $n$ and $\Pr[A>\b]^{-1}$, respectively. Furthermore, when $n \Pr[A>\b] = C$ the formulas are asymptotically identical as $\ell(n) = \ell(C\Pr[A>\b]^{-1}) \sim \ell(\Pr[A>\b]^{-1})$ as $n \rightarrow \infty$. 
\end{remark}

\begin{remark}
Note that most well known examples of slowly varying functions satisfy the condition $\ell'(x) x / \ell(x) \rightarrow 0$ as $x \rightarrow \infty$, including $\log^\beta x, \log^\beta(\log x), \beta>0,  \exp(\log x/ \log \log x), \exp(\log^\gamma x)$, for $0<\gamma <1$ [see Section 1.3.3 on p.16 in \cite{BG87}].
\end{remark}

\begin{proof}
Recall that
\begin{align} \label{eq:thm4-1}
\nonumber  \Pr[N_\b > n] & = \E [1-\bar{G}(L_\b)]^n \\
\nonumber & =   \int_{0}^\b{\left(1-\bar{G}(x) \right)^n \frac{dF(x)}{F(\b)}} \\
  & =   \int_{0}^{x_0}{\left(1-\bar{G}(x) \right)^n \frac{dF(x)}{F(\b)}}   + \int_{x_0}^\b{\left(1-\bar{G}(x) \right)^n \frac{dF(x)}{F(\b)}} .  
\end{align}
Now, given that $\ell(x)$ is eventually differentiable ($x \geq x_0$) and slowly varying with $\ell'(x) x/\ell(x) \rightarrow 0$ as $x \rightarrow \infty$, it follows that $  d\bar{F}(x) =(1+o(1))  \alpha \bar{G}(x)^{\alpha - 1}   \ell^{-1}(1/\bar{G}(x))  d\bar{G}(x) $ as $x \rightarrow \infty$. Thus, for any $0<\epsilon <1$, we can select $x_0$ large enough such that
\begin{align}\label{eq:thm4-2-UB}
\nonumber \Pr[N_\b > n]  & \leq  \left(1-\bar{G}(x_0) \right)^n -(1+\epsilon)^{1/2}\int_{x_0}^\b{(1-\bar{G}(x))^n  \frac{\alpha \bar{G}(x)^{\alpha - 1}   d\bar{G}(x) }{\ell(1/\bar{G}(x)) F(\b)}} \\
 &  =   \left(1-\bar{G}(x_0) \right)^n +  (1+\epsilon)^{1/2} \int_{\bar{G}(\b)}^{\bar{G}(x_0)}{(1-z)^n  \frac{\alpha z^{\alpha - 1}   dz }{\ell(1/z) F(\b)}}, 
\end{align}
which follows from the absolute continuity of $G(x)$, i.e., $\bar G(A)$ is uniformly distributed in [0,1]. 

Now, for $\alpha \geq 1$, we consider two different cases: (a) $ n \bar G (\b)  \geq  \log n $ and (b) $C < n \bar G (\b) <  \log n$.

\noindent Case (a):  $ n \bar G (\b)  \geq  \log n $. Observe that, for any fixed $H>\alpha + 6$, we can make $H \bar G (\b)$ small enough by picking $b_0$ large. Now, by continuity of $G(x)$, there exists $x_0$ such that $\bar{G}(x_0) = H  \bar{G}(\b)$; we can choose $x_0$ larger than in \eqref{eq:thm4-2-UB} by picking $b_0$ large enough. Next, using the elementary inequality $1-x \leq e^{-x}, x \geq 0$, we upper bound the preceding expression by
\begin{align} \label{eq:thm4-2-a1}
 \nonumber \Pr[N_\b > n]  & \leq  e^{-n\bar{G}(x_0)} +  \frac{\alpha(1+\epsilon)^{1/2}}{F(\b)}  \int_{\bar{G}(\b)}^{H \bar{G}(\b)}{(1-z)^n    \frac{ z^{\alpha - 1}   dz }{\ell(1/z) }} \\
 \nonumber & \leq  e^{-nH \bar{G}(\b)} +  \frac{\alpha(1+\epsilon)^{1/2}  }{\ell(1/ \bar{G}(\b)) F(\b)} \sup_{\bar G (\b) \leq z \leq H \bar G (\b)} \frac{\ell(1/\bar{G} (\b))}{\ell(1/z )}  \int_{\bar{G}(\b)}^{H \bar{G}(\b)}{(1-z)^n   z^{\alpha - 1} dz} \\
 \nonumber & \leq  e^{-nH \bar{G}(\b)} +  \frac{\alpha(1+\epsilon)}{\ell(1/ \bar{G}(\b)) F(\b)}  \int_{\bar{G}(\b)}^{H \bar{G}(\b)}{(1-z)^n   z^{\alpha - 1} dz} \\
 & \eqdef I_0 + I_1,
\end{align}
where, for the third inequality, by the uniform convergence theorem (see \cite{BG87}) of $\ell(x)$, $\bar G (\b)^{-1}$ can be chosen large enough such that $ \sup_{  (H\bar G(\b))^{-1} \leq y \leq  \bar G(\b)^{-1}} \ell( \bar G(\b)^{-1}) / \ell(y) \leq (1+\epsilon)^{1/2}$. 

Now, we derive a lower bound for $I_1$ in \eqref{eq:thm4-2-a1}. Using the monotonicity of $z^{\alpha-1}, \alpha \geq1$ and since $F(\b) \leq 1$, we obtain
\begin{align*}  
 I_1  &  \geq  \frac{1}{\ell(1/ \bar{G}(\b)) }  \int_{\bar{G}(\b)}^{H \bar{G}(\b)}{(1-z)^n   z^{\alpha - 1} dz} \\
 &  \geq  \frac{1}{\bar{G}(\b)^{-\epsilon} } \bar{G}(\b)^{\alpha - 1}  \int_{\bar{G}(\b)}^{H \bar{G}(\b)}{(1-z)^n   dz} \\
  & =   \frac{1}{n+1} \bar{G}(\b)^{\alpha - 1 + \epsilon}  (1-\bar G (\b))^{n+1}  \left(  1 -  \left(\frac {1-H \bar G (\b))}{1 - \bar G (\b)} \right)^{n+1} \right), 
 \intertext{where in the second inequality, we use the property of slowly varying functions $\ell(x) \leq x^\epsilon$ for $x$ large enough. Now, observe that for all $x \geq 0 $ small enough, $1-x \geq e^{-2 x}$, yielding}
 I_1 & \geq   \frac{1}{4 n} \bar{G}(\b)^{\alpha}  e^{-4 n \bar G (\b)}  ,
\end{align*}
where the last inequality follows from the fact that $n/ (n+1) \geq 1/2 $ for $n\geq 1$ and $\bar G(\b)^{\alpha - 1 + \epsilon} \geq \bar G (\b)^{\alpha}$ since $\epsilon <1$. We also note that
 \begin{displaymath}
 \left( \frac{ 1-H\bar{G(\b)}}{1 - \bar G (\b)} \right)^{n+1} \leq \left(e^{-H\bar G (\b)}/ e^{-2 \bar G (\b)} \right)^{n} = e^{-(H-2) n \bar{G} (\b)} \leq e^{-(H-2) C} \leq 1/2,
 \end{displaymath}
where we use our assumption $n\bar G (\b) > C$ and choose $H$ large enough. Finally, we obtain
\begin{equation} \label{eq:I1-lb-1}
 I_1  \geq   \frac{1}{4 n} \bar{G}(\b)^{\alpha}  e^{-4 n \bar G (\b)}.
 \end{equation}

Now, we proceed with proving that $I_0/ I_1$ in \eqref{eq:thm4-2-a1} is negligible as $n \rightarrow \infty$. To this end, observe that 
\begin{align*} 
  \frac{I_0}{ I_1} & \leq 4 \frac{ e^{- H n \bar{G}(\b)}n }{ \bar G(\b)^{\alpha } e^{-4 n \bar G(\b)} }  \leq  4  \frac{ e^{- (H-4) n \bar G(\b)}n^{\alpha+1} }{ (n \bar G (\b))^{\alpha  }  } \\
    & \leq 4 e^{-(H - 4)n \bar G(\b)}n^{\alpha+1}     \leq  4 e^{-(\alpha + 2) \log n }n^{\alpha+1}, 
 \end{align*}
 where we use our assumption that $n \bar G (\b) \geq \log n > 1$ for $n > 2$, whereas for the last inequality, we also use the fact that $H  > \alpha + 6$. Thus, the preceding expression is upper bounded by
 \begin{equation} \label{eq:I0-1-a}
  \frac{I_0}{ I_1}   \leq \frac{4} {n} \leq \epsilon,    
 \end{equation}
for all $n \geq 4/ \epsilon$.

Now, we upper bound $I_1$ in \eqref{eq:thm4-2-a1} by changing the variables $z= 1-e^{-u/n}$,
\begin{align} \label{eq:I1-2}
\nonumber I_1 & =  \frac{\alpha(1+\epsilon)}{F(\b) \ell(1/\bar{G}(\b))}   \int_{-n \log \left(1-\bar{G}(\b)\right)}^{-n \log(1-H \bar G (\b))}{ \frac{e^{-u(n+1)/n}(1-e^{-u/n})^{\alpha-1}}{n} du}  \\
\nonumber & \leq   \frac{\alpha(1+\epsilon)}{F(\b) \ell(1/\bar{G}(\b))} \int_{-n\log \left(1-\bar{G}(\b)\right)}^{\infty}{  \frac{e^{-u}(1-e^{-u/n})^{\alpha-1}}{n} du},
\intertext{where for the inequality we use $e^{-u/n} \leq 1$ and extend the integral to infinity. Thus, for $\alpha \geq 1$, from the preceding expression using the inequality $1-e^{-x} \leq x,$ for $ x\geq 0$, we obtain}
\nonumber I_1 & \leq \frac{\alpha(1+\epsilon)}{F(\b) n \ell(1/\bar{G}(\b))}  \int_{-n \log \left(1-\bar{G}(\b)\right)}^{\infty} { e^{-u} \left(\frac{u}{n}\right)^{\alpha-1} du} \\
\nonumber & \leq   \frac{\alpha(1+\epsilon)}{F(\b) n^\alpha \ell(1/\bar{G}(\b))}   \int_{-n \log \left(1-\bar{G}(\b)\right)}^{\infty} { e^{-u} u^{\alpha-1} du}  \\
& =  \frac{\alpha(1+\epsilon)}{F(\b) n^\alpha \ell(1/\bar{G}(\b))}   \Gamma (-n \log \left(1-\bar{G}(\b)\right), \alpha) .
\end{align}

Combining \eqref{eq:I0-1-a} and \eqref{eq:I1-2}, we obtain for \eqref{eq:thm4-2-a1} that for all $n \geq n_0,b \geq b_0$,
\begin{align*} 
  \Pr[N_\b > n] & \leq  \frac{\alpha(1+\epsilon)^2}{F(\b) n^\alpha \ell(1/\bar{G}(\b))}   \Gamma (-n \log \left(1-\bar{G}(\b)\right), \alpha),
  \end{align*}
which completes the proof after replacing $(1+\epsilon)$ with $(1+\epsilon)^{1/2}$.


\noindent Case (b):  $ C <  n \bar G (\b)  <  \log n $. In this region, for any fixed $H>5$, we choose the smallest $m \geq 1$ such that $H e^m - 4 \geq (\alpha + 2) \log n / C$, i.e., $m =  \lceil \log \left((\alpha +2) \log n /C + 4 \right) - \log H  \rceil$. Furthermore, it is important to note that this choice of $m$ allows for $H e^m \bar{G}(\b) $ to be small enough, since $H e^m \bar{G}(\b)  \leq H e^m \log n / n = O(\log^2 n/n) \rightarrow 0$, as $n \rightarrow \infty$, by the assumption that $n \bar G (\b) < \log n$. Then, by continuity of $G(x)$, there exists $x_0$ such that $\bar{G}(x_0) = H e^m \bar{G}(\b)$ (larger than $x_0$ in \eqref{eq:thm4-2-UB} for $b_0$ large enough) and using the elementary inequality $1-x \leq e^{-x}, x \geq 0$, we upper bound the expression in \eqref{eq:thm4-2-UB} by
\begin{align*}\label{eq:thm4-2}
& \Pr[N_\b > n] \\
&   \leq  e^{-n\bar{G}(x_0)} +  \frac{\alpha(1+\epsilon)^{1/2}}{F(\b)} \left[ \int_{\bar{G}(\b)}^{H \bar{G}(\b)}{(1-z)^n    \frac{ z^{\alpha - 1}   dz }{\ell(1/z) }} +  \sum_{k=0}^{m-1} \int_{He^{k} \bar{G}(\b)}^{He^{k+1} \bar{G}(\b)}{e^{-nz}    \frac{ z^{\alpha - 1}   dz }{\ell(1/z)}} \right] \\
  &  \leq  e^{-n\bar{G}(x_0)} + \frac{\alpha(1+\epsilon)^{1/2}}{F(\b)} \left[ \frac{(1+\epsilon)^{1/2}}{\ell(1/  \bar G (\b))} \int_{\bar{G}(\b)}^{H \bar{G}(\b)}{(1-z)^n    z^{\alpha - 1}   dz  } \right.  \\ 
  & \left.  \hspace{45mm} +    \sum_{k=0}^{m-1} \frac{(1+\epsilon)^{1/2}}{\ell(1/ (e^{k}\bar G (\b)))} { (He^{k+1} \bar{G}(\b))^{\alpha - 1}   \int_{He^{k} \bar{G}(\b)}^{He^{k+1} \bar{G}(\b)}{e^{-nz}   } } \right], 
\end{align*}
where the last inequality follows from the monotonicity of $z^{\alpha - 1}$ for $\alpha \geq 1$ and the uniform convergence theorem of $\ell(x)$, $\sup_{ (H \bar G (\b))^{-1} \leq z \leq 1/\bar G (\b))} \ell(1/ \bar G (\b)) / \ell(z) \leq (1+\epsilon)^{1/2}$, while for the second term, note that $ \sup_{ (H e^{k+1} \bar G (\b))^{-1} \leq z \leq (H e^k\bar G (\b))^{-1}}  \ell(1/ (e^k\bar G (\b)) / \ell(z) \leq (1+\epsilon)^{1/2}$, $k = 0 \dots (m-1)$, since $H e^m \bar G (\b)$ is small enough. Now, since $\ell(x)/ \ell(x/e) \leq e$ for $x$ large, it follows that
\begin{align}
\nonumber \Pr[N_\b > n]    &  \leq  e^{-n H e^m \bar{G}(\b)} + \frac{\alpha(1+\epsilon)}{F(\b)\ell(1/\bar{G}(\b))} \int_{\bar{G}(\b)}^{H \bar{G}(\b)}{(1-z)^n    z^{\alpha - 1}   dz  }  \\
\nonumber  &   \hspace{25mm} + \frac{\alpha(1+\epsilon)}{F(\b)n \ell(1/\bar{G}(\b))}   \sum_{k=0}^{m-1} e^{k} {e^{-n He^k  \bar{G}(\b)}  (He^{k+1} \bar{G}(\b))^{\alpha - 1}    } \\
& \eqdef I_0 + I_1 + I_2.
\end{align}
Now, we derive a lower bound for $I_1$ following similar arguments as in \eqref{eq:I1-lb-1}. Note that, for $x \geq 0 $ small enough, $1-x \geq e^{-2x}$, and thus, for $H$ large enough, we have
\begin{align}\label{eq:I1-lb-2}
\nonumber \frac{\ell(1/\bar{G}(\b)) F(\b)}{\alpha(1+\epsilon)}  I_1  & \geq   \bar G(\b)^{\alpha -1 } \int_{\bar{G}(\b)}^{H \bar{G}(\b)}{(1-z)^n       dz  } \\
\nonumber & \geq   \bar G(\b)^{\alpha -1 }  \frac{(1-\bar G(\b))^{n+1} - (1-H\bar{G(\b)})^{n+1}}{n+1}  \\
 \nonumber & =   \bar G(\b)^{\alpha -1 } \frac{(1-\bar G(\b))^{n+1}}{n+1}  \left( 1 - \left( \frac{ 1-H\bar{G(\b)}}{1 - \bar G (\b)} \right)^{n+1} \right)\\
& \geq \frac{ \bar G(\b)^{\alpha -1 }}{ 4 n} e^{- 4 n \bar G(\b)},
\end{align}
 where the expression in brackets is bounded from below by $1/2$ as in \eqref{eq:I1-lb-1}. 
 
 Now, we prove that $I_0/ I_1$ in \eqref{eq:thm4-2} is negligible as $n \rightarrow \infty$. To this end, observe that 
\begin{align}\label{eq:I0-2} 
 \nonumber  \frac{I_0}{ I_1}  & \leq \frac{ 4 F(\b) }{\alpha(1+\epsilon)} \frac{ e^{- He^m n \bar{G}(\b)} \ell(1/\bar{G}(\b))n }{ \bar G(\b)^{\alpha -1 } e^{- 4 n \bar G(\b)} },\\
  \intertext{where we use \eqref{eq:I1-lb-2}. Next, since $ \alpha \geq 1, F(\b) \leq 1$, and using the standard property of slowly varying functions that $\ell(x) \leq x$ for large $x$ (see Theorem~1.5.6 on page 25 of \cite{BG87}), we obtain}
\nonumber  \frac{I_0}{ I_1} &  \leq 4  \frac{e^{- (He^m-4) n \bar{G}(\b)}n }{  \bar G(\b)^{\alpha  }  }, \\
 \intertext{and since $n \bar G (\b) > C$, we have }
\nonumber  \frac{I_0}{ I_1} & \leq  4  \frac{ e^{- (He^m-4)n \bar G (\b) }n^{\alpha+1} }{ (n \bar G (\b) )^{\alpha }  } \\
   \nonumber & \leq  4 C^{-\alpha } e^{-(He^m - 4)C}n^{\alpha+1} \\
  \nonumber    & \leq  4 C^{-\alpha } e^{-(\alpha + 2 ) \log n }n^{\alpha+1}, 
      \intertext{where the last inequality follows from the fact that $m$ was chosen so that $(H e^m - 4) \geq  (\alpha + 2 )\log n/ C $. Thus, the preceding expression can be  rewritten as}
      \frac{I_0}{ I_1}  & \leq   \frac{4} {C^{\alpha } n} \leq \epsilon/2,    
  \end{align}
  for all $n \geq 8 C^{-\alpha }/ \epsilon$.
 
 Next, for the ratio $I_2/I_1$ we proceed similarly as before
 \begin{align} \label{eq:I2-2} 
 \nonumber  \frac{I_2}{I_1} & =  \frac{4  \sum_{k=0}^{m-1} e^{k} {e^{-n He^k  \bar{G}(\b)}  (He^{k+1} \bar{G}(\b))^{\alpha - 1}    } } {  \bar G(\b)^{\alpha-1 } e^{-4 n \bar G(\b)} } \\
\nonumber &   \leq  4  \sum_{k=0}^{m-1} e^{k} {e^{- (He^k-4)  n \bar{G}(\b)}  (He^{k+1})^{\alpha - 1} }, \\
\nonumber &   \leq  4 H^{\alpha - 1}  \sum_{k=0}^{m-1} e^{k} {e^{-(He^k-4)  C}  e^{\alpha(k+1) - k - 1} } \\
&   \leq 4  H^{\alpha-1} e^{-HC} \sum_{k=0}^{\infty} {e^{-5(e^k-1)  C + \alpha (k+1)-1} } \leq \epsilon/2,
  \end{align}
 where for the last inequality we use $H>5$. Now, we further observe that the preceding sum is finite and  thus, letting $H \rightarrow \infty$, the above ratio converges to 0, i.e., $I_2 \leq \epsilon I_1/2$ for large $H$.


Hence, since the upper bound for $I_1$ from \eqref{eq:I1-2} holds in this case as well, by putting \eqref{eq:I0-2} and \eqref{eq:I2-2} together, we obtain for \eqref{eq:thm4-2} that for all $n \geq n_0, b \geq b_0$,
\begin{align*} 
  \Pr[N_\b > n] & \leq  \frac{\alpha(1+\epsilon)^2}{F(\b) n^\alpha \ell(1/\bar{G}(\b))}   \Gamma (-n \log \left(1-\bar{G}(\b)\right), \alpha), 
  \end{align*}
which completes the proof after replacing $(1+\epsilon)$ with $(1+\epsilon)^{1/2}$.

Last, we prove the lower bound for $n \bar G(\b) >C$; here, we do not need to distinguish two cases. Thus, starting from \eqref{eq:thm4-1} and proceeding with similar arguments as in the proof for the upper bound, we obtain
\begin{align*}
 \Pr[N_\b > n]&  \geq -(1-\epsilon)^{1/2} \int_{x_0}^\b{(1-\bar{G}(x))^n  \frac{\alpha \bar{G}(x)^{\alpha - 1}   d\bar{G}(x) }{\ell(1/\bar{G}(x)) F(\b)}} \\
 &  =   (1- \epsilon)^{1/2} \int_{\bar{G}(\b)}^{H \bar{G}(\b)}{(1-z)^n  \frac{\alpha z^{\alpha - 1}   dz }{\ell(1/z) F(\b)}} \\
 & \geq   \frac{\alpha(1-\epsilon)}{F(\b)  \ell(1/\bar{G}(\b))}    \int_{-(n+1) \log \left(1-\bar{G}(\b)\right)}^{-(n+1)  \log{(1-H \bar G(\b))}}{ \frac{e^{-u}(1-e^{-\frac{u}{n+1}})^{\alpha-1}}{n+1} du},
 \intertext{where we use the uniform convergence theorem of slowly varying functions (Theorem~1.2.1 on page 6 of \cite{BG87}) and pick $x_0 < \b$ such that $\bar G (x_0) = H \bar G (\b)$. Next, using the inequality $1-e^{-x} \geq (1-\delta) x$, for some $\delta >0$ and all $x \geq 0$ small enough, we have}
  \Pr[N_b > n -1]  & \geq \frac{\alpha(1-\epsilon) (1-\delta)^{\alpha-1}}{F(\b) \ell(1/\bar{G}(\b)) n}   \int_{-n \log (1-\bar{G}(\b))}^{-n  \log (1-H \bar{G}(\b))} { e^{-u} \left(\frac{u}{n}\right)^{\alpha-1} du} \\
&  \geq \frac{\alpha(1-\epsilon)^2}{F(\b) n^\alpha \ell(1/\bar{G}(\b))}   \int_{-n \log \left(1-\bar{G}(\b)\right)}^{H n  \bar{G}(\b)} { e^{-u} u^{\alpha-1} du}\\
& =  \frac{\alpha(1-\epsilon)^2}{F(\b) n^\alpha \ell(1/\bar{G}(\b))}  \left[ \int_{-n \log \left(1-\bar{G}(\b)\right)}^{\infty} { e^{-u} u^{\alpha-1} du} -   \int_{Hn  \bar{G}(\b)}^{\infty} { e^{-u} u^{\alpha-1} du} \right] \\
& \eqdef I_1 - I_2,
\end{align*}
where, in the second inequality, we choose $\delta >0$ small enough such that $(1-\delta)^{\alpha - 1} \geq (1- \epsilon)$ and note that $ - n \log (1- H\bar G (\b)) \geq  H n \bar G (\b)$. 
Next, we proceed with showing that $I_2/ I_1$ is negligible for large $n$. Note that $ -n \log (1- \bar G (\b)) \leq 2 n \bar G(\b)$, which follows from the elementary inequality $e^{-2x} \leq 1-x $ for all $x \geq 0$ small enough. Thus, 
\begin{align*}
\frac{I_2}{I_1} & \leq  \frac{ \int_{ H n \bar{G}(\b)}^{\infty} { e^{-u} u^{\alpha-1} du}}{ \int_{2n \bar{G}(\b)}^{\infty} { e^{-u} u^{\alpha-1} du}} \\
& \leq  \frac{ 2 (H n  \bar{G}(\b))^{\alpha - 1} e^{-H n \bar G (\b)}}{(2n\bar G(\b))^{\alpha - 1} e^{-2n \bar G (\b)}}, \\
\intertext{where we use the approximation for the incomplete gamma function for large $H$ [see Remark~\ref{re:2} of Theorem~\ref{thm:3}]. Now, using the main assumption $n \bar G (\b) > C$, we obtain }
\frac{I_2}{I_1}  & \leq   2 H^{\alpha - 1} e^{-(H-2) C}  \leq \epsilon,
\end{align*}
for $H$ large enough. 
Then, using the preceding observation, we obtain 
\begin{align*}
  \Pr[N_\b > n] & \geq  \frac{\alpha (1-3\epsilon ) }{n^\alpha\ell(1/\bar{G}(\b))} \Gamma (-n \log \left(1-\bar{G}(\b)\right), \alpha),
\end{align*}
which completes the proof after replacing $\epsilon$ with $\epsilon/3$. 

Now, if $\alpha  <1$, the proof uses almost identical arguments coupled with the fact that $u^{\alpha -1}$ is a decreasing function. We omit the details to avoid unnecessary repetitions. 
\qed \end{proof}

\begin{remark}\label{rem:end}
From the preceding two theorems we observe that $ \Pr[N_\b > n]$ behaves as a true power law of index $\alpha$ when $- n \log  \Pr[A \leq \b]  \rightarrow c$, $0 \leq c<\infty$, and has an exponential tail (geometric) when $n \Pr[A>b]  \rightarrow \infty$ ($n >> \Pr[A>b]^{-1}$). More specifically:\\
(i) If $-n \log \Pr[A \leq \b] \rightarrow c$, then by Theorem~\ref{thm:3}, as $n   \rightarrow \infty, n \Pr[A>b]  \rightarrow c$,
\begin{align*}
\Pr[N_\b > n] & \sim \frac{\alpha}{\ell(n) n^{\alpha}} \Gamma (c, \alpha).
\end{align*}
(ii) If $n \Pr[A>b]  \rightarrow \infty$, then $-n\log \Pr[A \leq \b] \rightarrow \infty$ and thus, as $n   \rightarrow \infty, \b  \rightarrow \infty, n \Pr[A>b]  \rightarrow \infty$,
\begin{align*}
 \Pr[N_\b > n] &  \sim  \frac{\alpha}{\ell(1/\bar{G}(\b))n} \bar{G}(\b)^{\alpha-1}  (1-\bar{G}(\b))^n, 
 \end{align*}
 which follows from Theorem~\ref{thm:4} and the asymptotic expansion of the Gamma function (see Remark \ref{re:2} of Theorem \ref{thm:3}).
\end{remark}

Interestingly, one can compute the distribution of $\Pr[N_{\b} > n] $ exactly for the special case when the parameter $\alpha$ takes integer values and $\ell(x) \equiv 1$.
\begin{proposition} \label{prop:integer}
If $\Pr[L > x] = \Pr[A>x]^\alpha$, for all $x \geq 0$ and $\alpha$ is a positive integer, then
\begin{displaymath}
\Pr[N_{\b} > n] = \frac{1}{\Pr[L \leq \b]}\sum_{i = 1}^{\alpha} \frac{\alpha! \; n! \; \Pr[A>\b]^{\alpha - i}}{(\alpha - i)!(n+i)!} \Pr[A\leq \b]^{n+i}.
\end{displaymath}
\end{proposition}
\begin{proof}
It follows directly from \eqref{eq:thm4-1} using integration by parts.\qed \end{proof}

Finally, in the following proposition, we describe the tail of $\Pr[N_\b > n]$ for fixed and possibly small $\b$. This complements the conclusion of Remark~\ref{rem:end}(ii), however, we need $\ell(x) \equiv 1$. 

\begin{proposition} \label{prop:exptail}
Let $\b $ be fixed. If $\Pr[L>x] =  \Pr[A>x]^{\alpha}$, $\alpha >0$, $x \geq 0$, then
\begin{displaymath}
\Pr[N_{\b} > n] \sim \frac{\alpha \Pr[A>\b]^{\alpha - 1}} {\Pr[L \leq \b]} \frac{ \Pr[A \leq \b]^{n+1}} {n+1} \quad \text{as } n \rightarrow \infty.
\end{displaymath}
\end{proposition}

\begin{proof}
See Section 4 in \cite{JS2012}. 
\qed 
\end{proof}

\section{Simulation Experiments} \label{s:sim}
In this section, we illustrate the validity of our theoretical results with simulation experiments. In all of the experiments, we observed that our uniform exact asymptotics is literally indistinguishable from the simulation. In the following examples, we present the simulation experiments resulting from $10^8$ (or more) independent samples of $N_{\b,i}, 1 \leq i\leq 10^8$. This number of samples was needed to ensure at least 100 independent occurrences in the lightest end of the tail that is presented in the figures ($N_{\b,i} \geq n_{\max}$), thus providing a good confidence interval.

\begin{figure}[h]
\begin{minipage}[b]{0.5\linewidth}
\centering
\epsfig{file=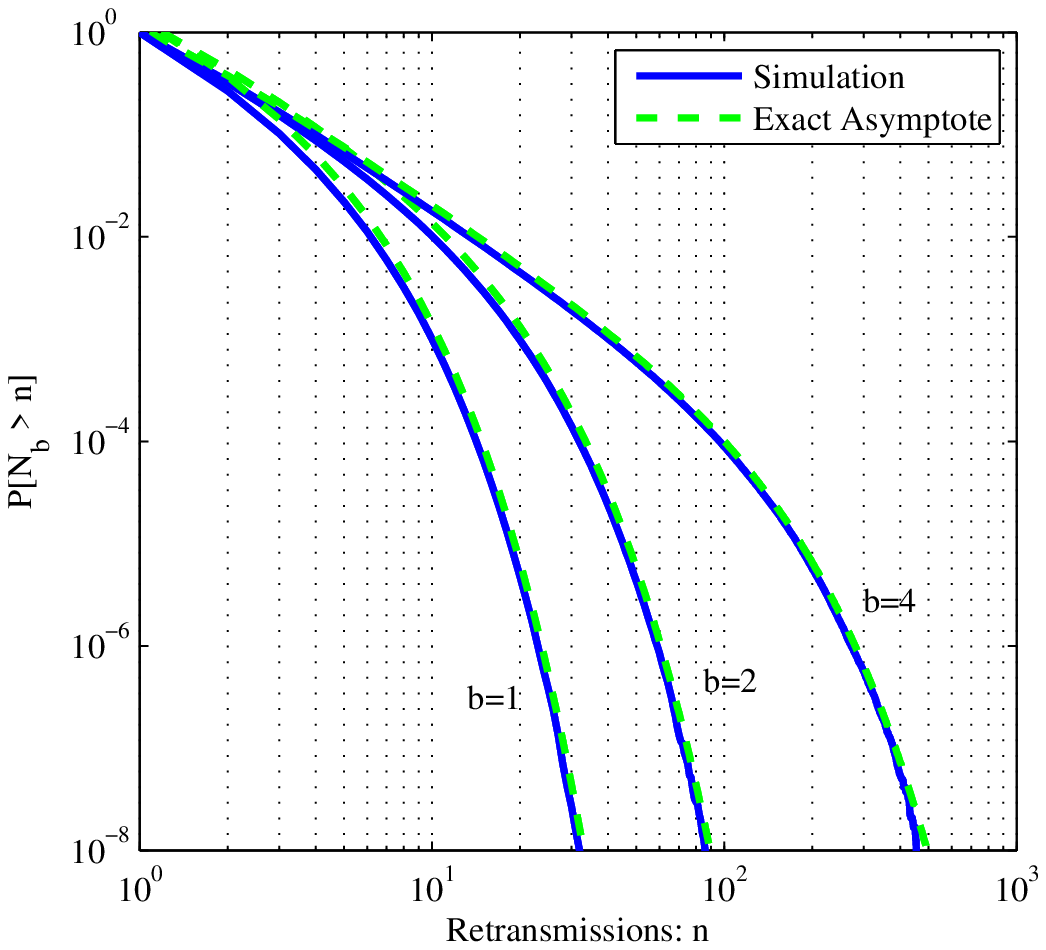, width = 2.8in, height = 2.4in}
\caption{\it Example 1(a). Exact asymptotics for $\alpha >1$.}
\label{fig:1}
\end{minipage}
\hspace{0.1cm}
\begin{minipage}[b]{0.5\linewidth}
\centering
\epsfig{file=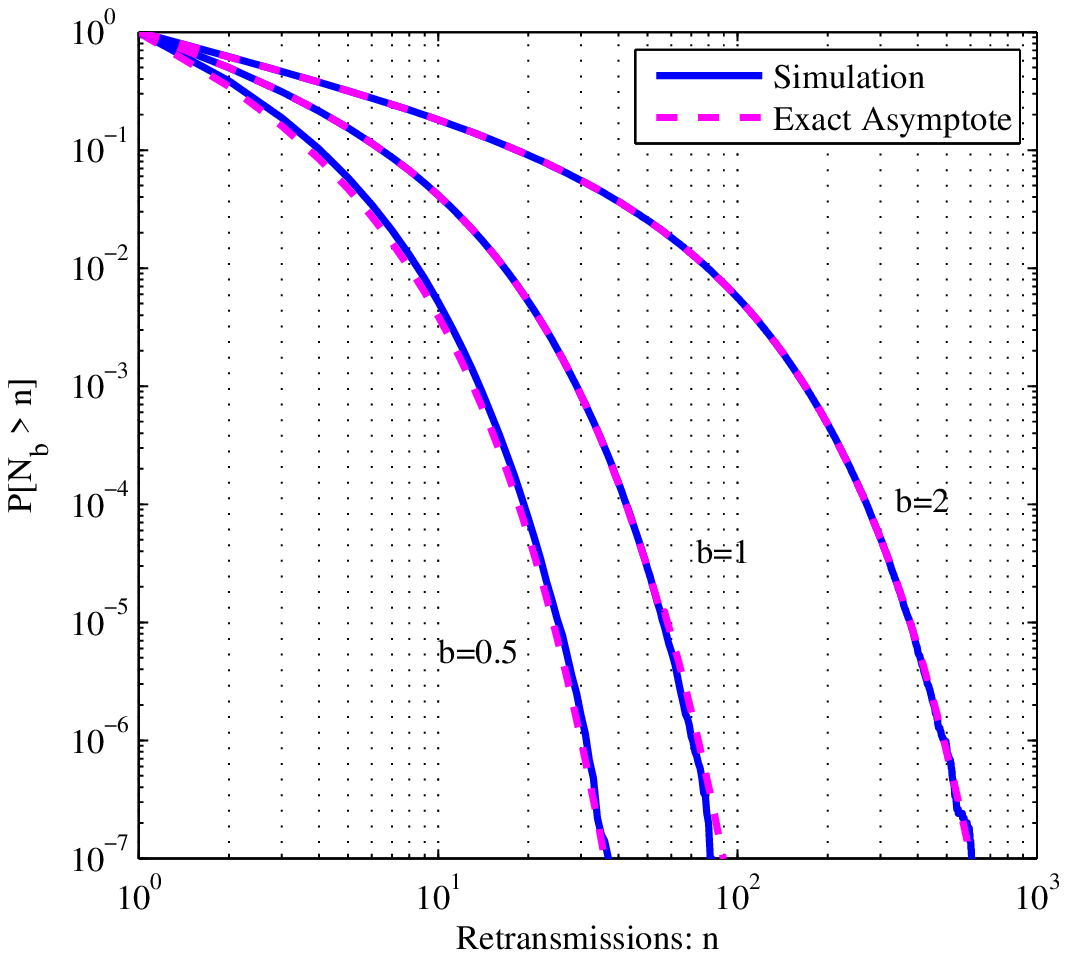, width = 2.8in, height = 2.35in}
\caption{\it Example 1(b). Exact asymptotics for $\alpha <1$.}
\label{fig:2}
\end{minipage}
\end{figure}

\textbf{Example 1. }This example illustrates the uniform exact asymptotics presented in Theorems~\ref{thm:3} and~\ref{thm:4}, i.e., approximation \eqref{eq:approx}, which combines the results from both theorems. We assume that $L$ and $A$ follow exponential distributions with parameters $\lambda=2$ and $\mu=1$, respectively. It is thus clear that $\bar{F}(x) = e^{-2x} = \bar{G}(x) ^{\alpha}$, where $\alpha =  2$ and $\ell(x) \equiv 1$. Now, approximation \eqref{eq:approx} states that $\Pr[N_\b > n]$ is given by $(1-e^{-2\b})^{-1} 2 n^{-2} \Gamma(ne^{-\b},2)$. Note that we added a factor $\Pr[L \leq \b]^{-1} = (1-e^{-2\b})^{-1} $, as in Propositions~\ref{prop:integer} and~\ref{prop:exptail}, for increased precision when $\b$ is small; we add such a factor to approximation~\eqref{eq:approx} in other examples as well. We simulate different scenarios when the data sizes $L_\b$ are upper bounded by $\b$ equal to 1, 2 and 4. The simulation results are plotted on log-log scale in Fig. \ref{fig:1}.

From Fig. \ref{fig:1}, we observe that the numerical asymptote approximates the simulation exactly for all different scenarios, even for very small values of $n$ (large probabilities). We further validate our approximation by considering scenarios where $L,A$ are exponentially distributed but $\alpha <1$; in fact, this case tends to induce longer delays due to larger average data size compared to the channel availability periods. In this case, we obtain $\alpha = 0.5$ by assuming $\lambda = 1$ and $\mu = 2$. Again, the simulation results and the asymptotic formulas are basically indistinguishable for all $n$, as illustrated in Fig.~\ref{fig:2}.

For both cases, we deduce that for $\b$ small the power law asymptotics covers a smaller region of the distribution of $N_\b$ and, as $n$ increases, the exponential tail becomes more evident and eventually dominates. As $\b$ becomes large - recall that $\b \rightarrow \infty$ corresponds to the untruncated case where the power law phenomenon arises -  the exponential tail becomes less distinguishable.

\begin{figure}[h]
\begin{minipage}[b]{0.5\linewidth}
\centering
\epsfig{file=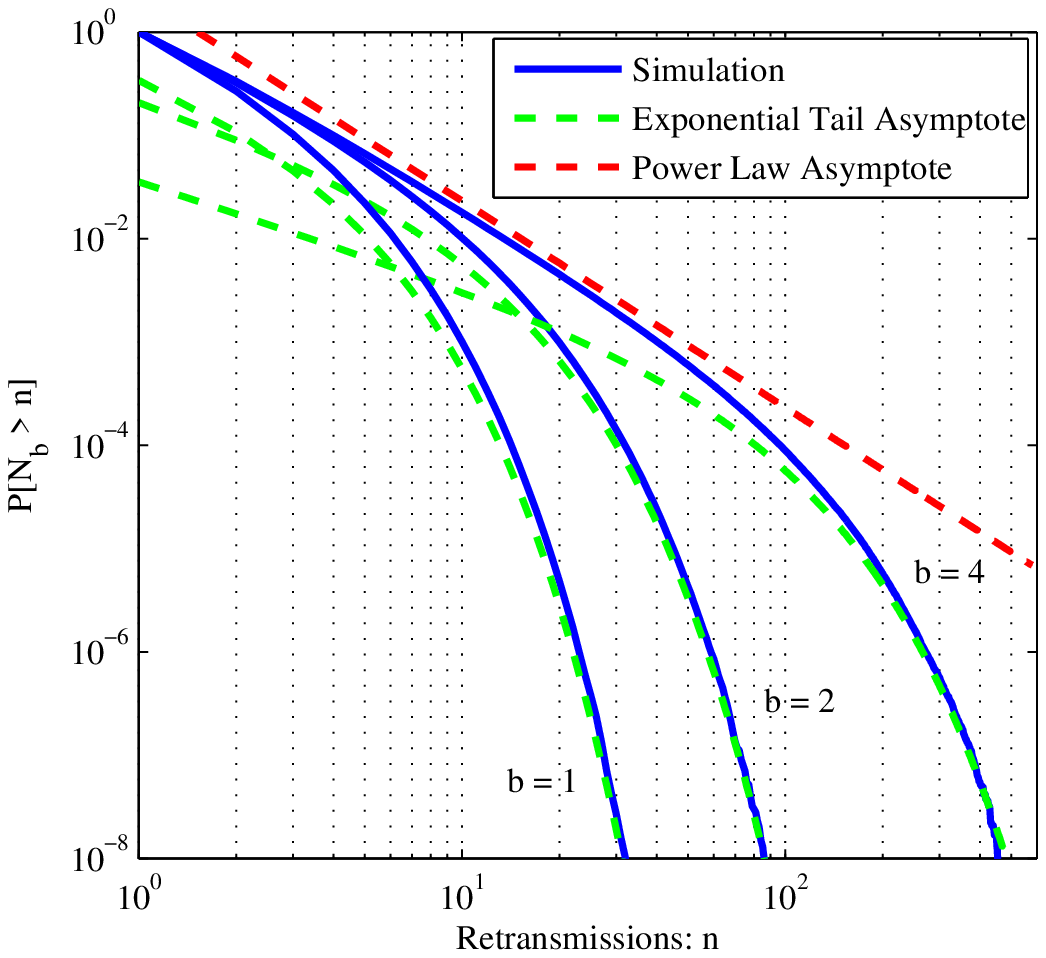, width = 2.8in, height = 2.4in}
\caption{\it Example 2. Power law versus exponential tail asymptotics.}\label{fig:3}
\end{minipage}
\hspace{0.2cm}
\begin{minipage}[b]{0.5\linewidth}
\centering
\epsfig{file=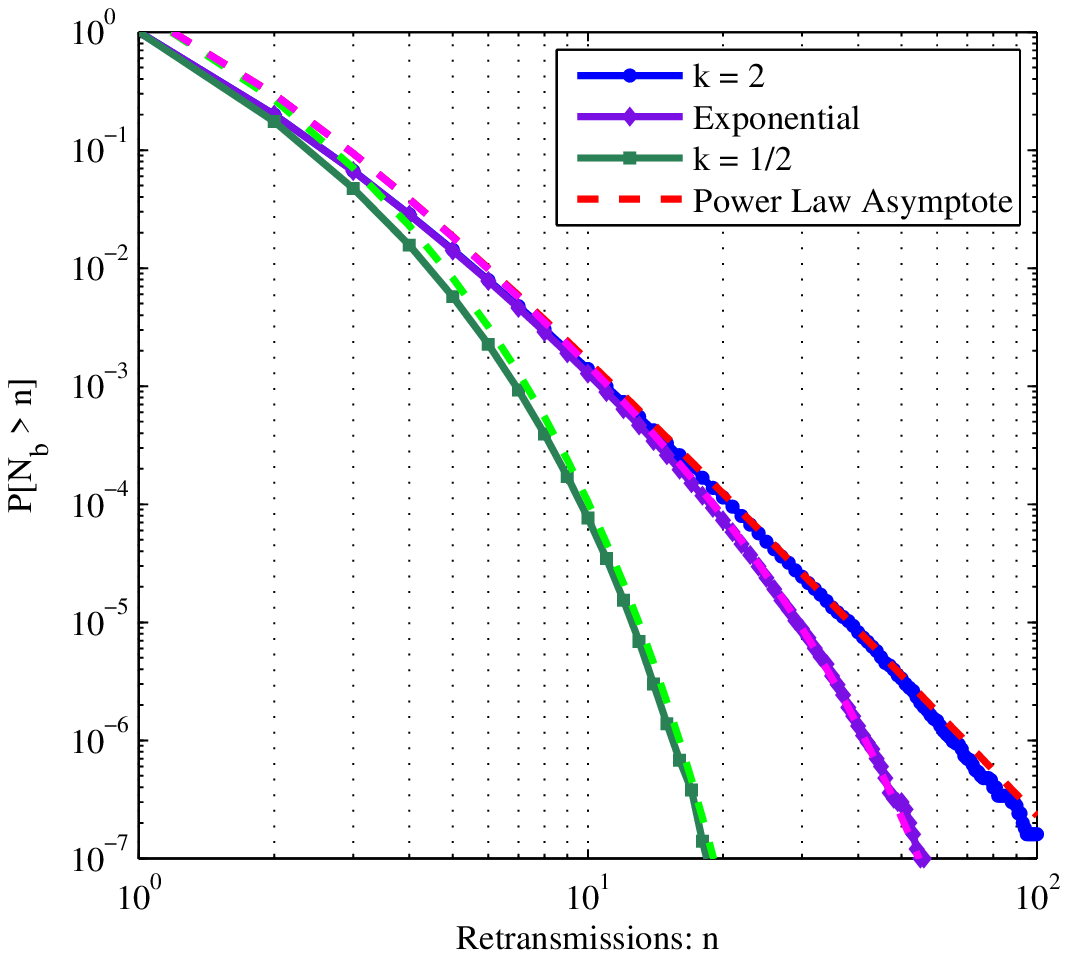 , width = 2.8in, height = 2.4in}
\caption{\it Example 3. Power law region increases for lighter tails of $L,A$.}
\label{fig:4}
\end{minipage}
\end{figure}


\textbf{Example 2. }This example demonstrates the exact asymptotics for the exponential tail as $n \rightarrow \infty$ and $\b$ is fixed, as in Proposition \ref{prop:exptail}. Note that this proposition gives the exact asymptotic formula for the region $n  \gg 1 / \bar{G}(\b)$ and lends merit to our Theorems~\ref{thm:2} and~\ref{thm:4}. Informally, we could say that a point $n_\b$ such that $-n_\b \log (1-\bar{G}(\b)) \approx n_\b \bar{G}(\b) = \alpha \log n_\b$ represents the transition from power law to the exponential tail. We assume that $L,A$ are exponentially distributed with $\lambda = 2$ and $\mu = 1$ (as in the first case of Example 1). Roughly speaking, we can see from Fig. \ref{fig:3} that the exponential asymptote appears to fit well starting from $n_\b \approx \alpha e^{\b}$, i.e., $n_\b \approx 6, 15, 100$ for $\b = 1, 2,4$, respectively.

\textbf{Example 3. }This example highlights the impact of the distribution type of channel availability periods $\bar{G}(x)=\Pr[ A > x]$. We consider some fixed $b$, namely $\b=8$, and assume that the matching between data sizes and channel availability, as defined in Theorems~\ref{thm:3} and~\ref{thm:4}, is determined by the parameter $\alpha=4$. We assume Weibull\footnotemark[1] distributions for $L,A$ with the same index $k$ and $\mu_L, \mu_A$ respectively, such that $\alpha = (\mu_A/\mu_L)^k$. The simulations include three different cases for the aforementioned distributions: Weibull with index $k=1 $ (exponential) where $\mu_L=1$ and $\mu_A =4$, Weibull (normal-like) with index $k=2 $ ($\mu_L= 1, \mu_A = 2$) and Weibull with $k = 1/2$ ($\mu_L = 1 , \mu_A =16$). Fig. \ref{fig:4} illustrates the exact asymptotics from equation \eqref{eq:approx}, shown with the lighter dashed lines; the main power law asymptote appears in the main body of all three distributions. We observe that heavier distributions (Weibull with $k=1/2$) correspond to smaller regions for the power law main body of the distribution $\Pr[N_\b >n ]$. On the other hand, the case with the lighter Gaussian like distributions for $k=2$ follows almost entirely the power law asymptotics in the region presented in Fig. \ref{fig:4}. This increase in the power law region can be inferred from our theorems, which show that the transition from the power law main body to the exponential tail occurs roughly at $n_\b \approx \bar{G}(\b)^{-1}$. Hence, the lighter the tail of the distribution of $A$, the larger the size of the power law region.

\footnotetext[1]{In general, a Weibull distribution with index $k$ has a complementary cumulative distribution function $\Pr[X>x] = e^{-(x/\mu)^k}$, where $\mu$ is the parameter that determines the mean. }

\begin{figure}[h]
\begin{minipage}[b]{0.5\linewidth}
\centering
 \epsfig{file=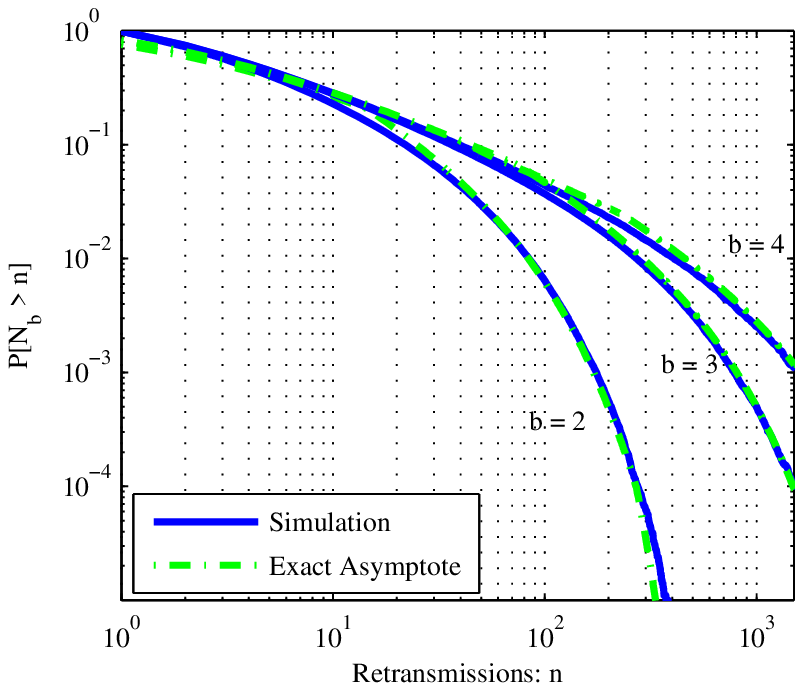 , width = 2.8in, height = 2.4in}
\caption{\it Example 4(a). Exact asymptotics for the case where $L$ is Gamma distributed.}
\label{fig:6}
\end{minipage}
\hspace{0.2cm}
\begin{minipage}[b]{0.5\linewidth}
\centering
\epsfig{file=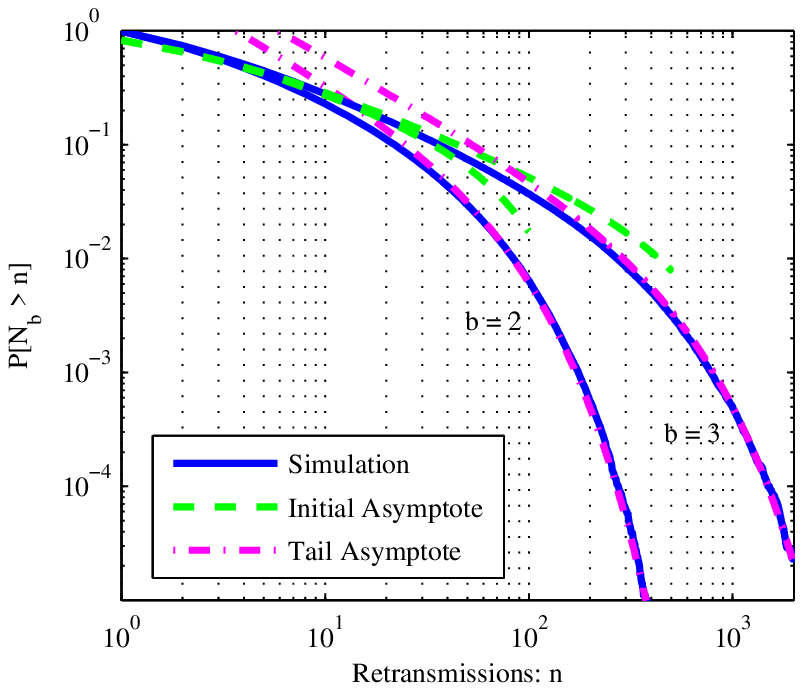 , width = 2.8in, height = 2.4in}
\caption{\it Example 4(b). The asymptotes from Theorems~\ref{thm:3} and~\ref{thm:4} for Gamma distributed $L$.}
\label{fig:7}
\end{minipage}
\end{figure}

\textbf{Example 4. }In this last example, we study the case where there is a more general functional relationship between the distributions of availability periods $A$ and data sizes $L$, as Theorems~\ref{thm:3} and \ref{thm:4} assume. In particular, we consider the case $ \bar{F}(x) =  \bar{G}(x)^\alpha / \ell(\bar{G}(x)^{-1}) $, where $\ell(x)$ is slowly varying. We validate the approximation \eqref{eq:approx} in this more general setting.

In particular, the availability periods $A$ are exponentially distributed with parameter $\mu$ while the data sizes $L$ follow the Gamma distribution with parameters ($\lambda, k$); the tail of the Gamma distribution function is defined as $ \lambda^{k} \Gamma(k)^{-1}   \int_x^{\infty} e^{-\lambda x} x^{k-1} dx = \Gamma(\lambda x, k) / \Gamma(k)$ and, therefore, the tail distribution of $L$ can be approximated by $\bar{F}(x) \sim (\lambda^{k-1}/ \Gamma(k))$ $ x^{k-1} e^{-\lambda x}$ for large $x$. 

We can easily verify that $\bar{F}(x) = f(\mu^{-1}\log{\bar{G}(x)^{-1}})  \bar{G}(x)^\alpha$, where $\alpha = \lambda / \mu$ and $$f(x) = \lambda^{k-1} \Gamma(k)^{-1} \int_0^{\infty} e^{-z} (z/\lambda + x)^{k-1} dz.$$ Hence, the slowly varying function in Theorems~\ref{thm:3} and~\ref{thm:4} is $\ell(x) =1/ f(\mu^{-1} \log x)$. From the preceding integral representation for $f(x)$, it can be easily shown that $\ell(x) \approx \Gamma(k)  \alpha^{1-k} \log^{1-k} x$, which is indeed slowly varying, and $$\bar{F}(x) \approx ( \alpha^{k-1} / \Gamma(k))  \log(\bar{G}(x)^{-1})^{k-1}  \bar{G}(x)^\alpha.$$ We take $\lambda = 2, k = 2$ and $\mu = 2$ and run simulations for $\b = \{2, 3, 4\}$. In Fig. \ref{fig:6}, we demonstrate the results using the approximation \eqref{eq:approx}. Interestingly, our analytic approximation works nicely even for small values of $n$ and $\b$ although the conditions in our theorems require $n$ and $\b$ to be large.

In Fig.~\ref{fig:7}, we elaborate on the preceding example. To this end, we plot two asymptotes: (i) the `Initial Asymptote' corresponding to the power law asymptote provided by Theorem~\ref{thm:3} and (ii) the `Tail Asymptote' from Theorem~\ref{thm:4}. Combining the two, we derive the approximation  \eqref{eq:approx}, as we have already shown in Fig.~\ref{fig:6}. Hereby, we see from Fig.~\ref{fig:7} that both asymptotes are needed to approximate the entire distribution well, i.e., the `Initial Asymptote' fits well the first part of the distribution, whereas the `Tail Asymptote' is inaccurate in the beginning but works well for the tail. Recall that these two asymptotes differ only in the argument of the slowly varying function $\ell(\cdot)$, which is equal to $n$ for the `Initial Asymptote' and $\bar{G}(\b)^{-1}$ for the tail.

\section{Proofs} \label{sec:4}
In this section, we present the proofs of Theorems~\ref{thm:1} and \ref{thm:2}.

\begin{proof}[of Theorem~\ref{thm:1}]
Without loss of generality, we assume $\epsilon <1$. We first observe that for $n \geq n_0$ the assumption $\Pr[A>\b] \leq 1 /n^{1+\epsilon}$ implies that $\b \geq \b_0$ for some $\b_0$; $\b_0$ can be made arbitrarily large by increasing $n_0$. Hence, the upper bound follows from Proposition \ref{prop:UB}. 

For the lower bound, we have
\begin{align*}
\Pr[N_\b > n]&= \E [1-\bar{G}(L_\b)]^n \\
& \geq  \E \left [(1-\bar{G}(L_\b))^n \textbf{1}(L_\b \geq x_0) \right] \\
& \geq (1 - \bar G(x_0))^n \Pr ( \b \geq L \geq x_0 ), 
\intertext{since $\Pr(L \leq \b) \leq 1$. Now, by our assumption $\bar G(\b) \leq  1/n^{1+\epsilon}$ and the continuity of $G(x)$, there exists $x_0 <  \b$ such that $\bar G(x_0) = 2 /n$. Hence, }
\Pr[N_\b > n] & \geq \left(1 - \frac{2 }{n} \right)^n  (\bar F(x_0)- \bar F(\b) ) \\
& \geq e^{-4 } \left( \bar G(x_0)^{\alpha (1+\epsilon/3)} -  \bar G (\b)^{\alpha (1-\epsilon/3)} \right), \\ 
\intertext{where we use our main assumption $ \bar G (x)^{\alpha(1+\epsilon/3)} \leq \bar F(x) \leq \bar G (x)^{\alpha(1-\epsilon/3)}$ and the inequality $1-x \geq e^{-2x}$ for small $x$. Now, by the assumption that $\bar G(\b) \leq 1/n^{1+\epsilon}$, we have}
\Pr[N_\b > n] & \geq e^{-4 } \left[ \left(\frac{2}{n} \right)^{\alpha (1+\epsilon/3)} - \left(\frac{1}{n^{1+\epsilon}} \right) ^{\alpha (1-\epsilon/3)} \right]. \\
\intertext{Next, observe that $(1+\epsilon) (1- \epsilon/3) > (1+\epsilon/3)$, since $\epsilon <1 $, and thus, }
 \Pr[N_\b > n]  & \geq  e^{-4 } \left( \frac{2^\alpha}{n^{\alpha (1+\epsilon/3)}}  - \frac{1}{n^{\alpha (1-\epsilon/3)}} \right) \\
&  \geq  e^{-4 }  \frac{ 2^ {\alpha}  - 1 } {n^ {\alpha(1+\epsilon/3)} } =    \frac{h_\epsilon } {n^ {\alpha(1+\epsilon/3)} },
\end{align*}
where we set $ h_\epsilon =   e^{-4 } ( 2^ {\alpha}  - 1)  $ and by taking the logarithm, we obtain 
 \begin{align*}
 \log \Pr[N_\b > n]& \geq  \log h_\epsilon   - \alpha(1+ \epsilon/3) \log n.
\end{align*}

Finally, since $\log n$ is increasing in $n$, we can choose $n_0$ such that for all $n \geq n_0$, $2 \alpha \epsilon \log n \geq  -3 \log h_\epsilon$, i.e.,
 \begin{align*}
 \log \Pr[N_\b > n]&\geq - \alpha(1+\epsilon) \log n. 
\end{align*}\qed
 \end{proof}

Last, we prove Theorem~\ref{thm:2}.

\begin{proof}[of Theorem~\ref{thm:2}]
First, we prove the result in the region $n \bar G(\b) \geq  \delta \log n$, for any fixed $\delta >0$, where the distribution of $\Pr[N_\b >n]$ approaches a geometric tail. 

Note that, for all $\b \geq \b_0$, the upper bound follows from Proposition~\ref{prop:UB}. For the lower bound, observe that 
\begin{align*}
\Pr[N_\b > n] &=  \E [1-\bar{G}(L_\b)]^n \\
& =   \int_{0}^\b{\left(1- \frac{ \bar{F}(x)^{\frac{1}{\alpha}}}{  \ell^{1/\alpha}(\bar{G}(x)^{-1}) } \right)^n \frac{dF(x)}{F(\b)}} \\
& \geq   \int_{x_0}^\b \left(1- \frac{ \bar{F}(x)^{\frac{1}{\alpha}}}{  \ell^{1/\alpha}(\bar{G}(x)^{-1}) } \right)^n dF(x)  , \\
\intertext{by $F(\b) \leq 1$. Then, by the continuity of $G(x)$, we can choose $x_0$ such that $\bar G(x_0 ) = \lambda_\epsilon \bar G(\b)$ where $\lambda_\epsilon = \lambda/(1-\epsilon)$, and $\lambda > 2$ is such that $\bar G(x_0) \leq 1$. Now, recall the uniform convergence theorem of $\ell(x)$ (Theorem~1.2.1 on page 6 of \cite{BG87}) in the region $\{\bar{G}(\b)^{-1} , \lambda_\epsilon^{-1} \bar{G}(\b)^{-1}  \}$ and thus, for $\b_0$ large enough, $\b \geq \b_0$, }
\Pr[N_\b > n] &  \geq \int_{x_0}^\b \left(1- (1-\epsilon)\frac{ \bar{F}(x)^{\frac{1}{\alpha}}}{  \ell^{1/\alpha}(\bar{G}(\b)^{-1}) } \right)^n dF(x)   \\
& = \E \left[ \left(1- (1-\epsilon)\frac{ \bar{F}(L)^{\frac{1}{\alpha}}}{  \ell^{1/\alpha}(\bar{G}(\b)^{-1}) } \right)^n \textbf{1}\left( \bar F(\b)\leq  \bar F(L) \leq \bar F(x_0) \right) \right]\\
& \geq  \frac{\alpha \ell(\bar{G}(\b)^{-1})} {(1-\epsilon)^\alpha} \int_{ (1- \epsilon)  \bar F(\b)^{1/\alpha}/ \ell^{1/\alpha}(\bar{G}(\b)^{-1})} ^{\ (1- \epsilon)  \bar F(x_0)^{1/\alpha}/ \ell^{1/\alpha}(\bar{G}(\b)^{-1})} \left(1- z\right)^n z^{\alpha -1} dz, 
\intertext{which follows from $\bar F(L) = U$, where $U$ is uniformly distributed in $[0,1]$, and the change of variables $z = (1- \epsilon) U^{1/\alpha}/ \ell^{1/\alpha}(\bar{G}(\b)^{-1}) $. Thus, using our main assumption that $\bar F(\b)^{1/\alpha}=\ell^{1/\alpha}(\bar G(\b)^{-1})\bar G(\b) $, and by $\bar F(x_0)^{1/\alpha} = \ell^{1/\alpha}(\bar G(x_0)^{-1})\bar G(x_0) =  \ell^{1/\alpha}(\lambda_\epsilon^{-1} \bar G(\b)^{-1}) \lambda_\epsilon \bar G(\b)  \geq (1-\epsilon) \lambda_\epsilon \bar G(\b) \ell^{1/\alpha}(\bar G(\b)^{-1}) = \lambda \bar G(\b)  \ell^{1/\alpha}(\bar G(\b)^{-1}) $, we obtain }
\Pr[N_\b > n] & \geq  \frac{\alpha \ell(\bar{G}(\b)^{-1})} {(1-\epsilon)^\alpha} \int_{ (1- \epsilon)  \bar G(\b)} ^{ (1- \epsilon) \lambda \bar G(\b)} \left(1- z\right)^n z^{\alpha -1} dz. 
\end{align*}

Now, if $\alpha \geq 1$, $z^{\alpha-1}$ is monotonically increasing and thus
\begin{align*}
 \Pr[N_\b > n] & \geq \frac{\alpha  \ell(\bar{G}(\b)^{-1})  (1-\epsilon)^{\alpha-1}\bar{G}(\b)^{\alpha-1}}{(1-\epsilon)^\alpha } \int_{ (1-\epsilon) \bar{G}(\b)}^{\lambda (1-\epsilon) \bar{G}(\b)}  (1-z)^n  dz \\
 & \geq\alpha  \bar{G}(\b)^{\alpha+\alpha \epsilon-1} \left. \frac{(1-z)^{n+1}}{n+1} \right|^{ (1-\epsilon) \bar{G}(\b)}_{\lambda (1-\epsilon) \bar{G}(\b)} 
 \end{align*}
 \begin{align*}
 & = \frac{\alpha \bar{G}(\b)^{\alpha(1+ \epsilon)-1} }{ n+1} 
\bigg[ \left(1-(1-\epsilon) \bar{G}(\b)\right)^{n+1}  \left. - (1-\lambda (1-\epsilon) \bar{G}(\b))^{n+1} \right] \\
 & \geq \frac{   \bar{G}(\b)^{\alpha(1+ \epsilon)-1}}{2 n}  \left(1-(1-\epsilon) \bar{G}(\b)\right)^{n+1} \left[1- \left( \frac{1-\lambda(1-\epsilon) \bar{G}(\b) }{1-(1-\epsilon) \bar{G}(\b)}\right)^{n+1} \right],
\end{align*}
where, for the second inequality, we use the standard property of slowly varying functions (see Theorem~1.5.6 on page 25 of \cite{BG87}) that $ \ell(x) \geq  x^{-\alpha \epsilon}$ for large $x$, and thus $\ell(\bar{G}(\b)^{-1}) \geq   \bar{G}(\b)^{\epsilon \alpha} $, whereas the last inequality is implied by $\alpha \geq 1$ and  $n+1 \leq 2 n$ for $n >1$. 

Next, using the inequalities $ e^{-2x} \leq (1-x) \leq e^{-x}$ for $x$ small enough, we have
\begin{align*}
  \left( \frac{1-\lambda (1-\epsilon) \bar{G}(\b) }{1-(1-\epsilon) \bar{G}(\b)}\right)^{n+1}  &\leq  \frac{e^{-\lambda(1-\epsilon) (n+1)\bar{G}(\b)}}{e^{- 2(1-\epsilon)(n+1) \bar{G}(\b)}}
 \leq  e^{- (\lambda-2)(1-\epsilon)n \bar{G}(\b)} \\ 
 & \leq  e^{- (\lambda-2)(1-\epsilon) \delta \log n} \leq 1/2,
\end{align*}
for all $n \geq n_0$, since we choose $\lambda>2$ and, by assumption, $n\bar{G}(\b) \geq  \delta \log n$. Therefore,
\begin{align*}
\Pr[N_\b > n] & \geq \frac{   \bar{G}(\b)^{\alpha(1+ \epsilon)-1} }{4n} \left(1-(1-\epsilon) \bar{G}(\b)\right)_.^{n+1}
\end{align*}
Then, by taking the logarithm,
\begin{align*}
\log \Pr[N_\b > n]  \geq &  - \log 4+ (\alpha(1+\epsilon)-1)\log \bar{G}(\b) + (n+1)\log (1-(1-\epsilon) \bar{G}(\b)) - \log n  \\
\geq &  -  \alpha \epsilon \log n  - (\alpha(1+\epsilon)-1)\log n  + (n+1)\log (1-(1-\epsilon) \bar{G}(\b)) - \log n ,
\intertext{where in the last inequality we used the assumption $\log \bar{G}(\b) \geq \log (\delta \log n)   - \log n \geq -\log n $, for $n \geq e^{e/\delta},$ and that $ \log 4 \leq \alpha \epsilon \log n $ for large $n$. Hence,}
\log \Pr[N_\b > n]  \geq &  - \alpha (1+2 \epsilon) \log n + (n+1) \log (1-(1-\epsilon) \bar{G}(\b)) \\
 \geq &  - \alpha (1+2\epsilon) \log n + n (1+2\epsilon) \log (1- \bar{G}(\b))\\
 = &  -(1+2\epsilon) \left[  \alpha  \log n - n \log (1- \bar{G}(\b)) \right].
 \end{align*}
Finally, dividing by $( \alpha \log n -n\log(1-\bar{G}(\b)) ) > 0$ and replacing $\epsilon$ with $\epsilon / 2$ yields
\begin{align*}
\frac {\log \Pr[N_\b > n]}{ \alpha \log n -n\log(1-\bar{G}(\b)) }  & \geq   -(1+\epsilon). 
\end{align*}

Symmetric arguments hold for the case where $\alpha < 1$. We omit the details.

Next, we prove the result for $n \bar G(\b) \leq \delta \log n$, for $n \geq n_0$, $\delta>0$. Note that the proof assumes any fixed $\delta >0$, and thus we can set  $\delta = \delta_\epsilon = \epsilon \alpha/ 4$. This assumption implies that $\b $ is large, say $\b \geq \b_0$, when $n$ is large, and thus the upper bound follows from Proposition~\ref{prop:UB}. For the lower bound, it is sufficient to prove \eqref{eq:thm1-a} instead of \eqref{eq:thm2-a} since for $\b$ large enough, $\alpha \log n \leq -n \log(1-\bar{G}(\b))+ \alpha \log n  \leq \alpha(1 +  \epsilon) \log n$. Without loss of generality, we may assume that $\epsilon <1/3$. Therefore, 
\begin{align*}
\Pr[N_\b > n]&= \E [1-\bar{G}(L_\b)]^n \\
& \geq  \E \left [(1-\bar{G}(L_\b))^n \textbf{1}(L_\b > x_0) \right] \\
& \geq (1 - \bar G(x_0))^n \Pr [\b \geq  L > x_0], \\
\intertext{and since $\bar G(\b) \leq \delta \log n/n$, by continuity of $G(x)$, there exist $x_0, x_1 (x_0 < x_1 \leq \b),$ such that $\bar G(x_0) = 2 \delta \log n /n$ and $\bar G(x_1) = \delta \log n /n$. Hence,}
\Pr[N_\b > n] & \geq \left(1 - \frac{2\delta \log n}{n} \right)^n \left[\bar F(x_0) - \bar F(x_1) \right] \\
& \geq e^{-4 \delta \log n}   \left[     \bar G (x_0)^{\alpha} \ell(\bar G(x_0)^{-1}) -   \bar G (x_1)^{\alpha} \ell(\bar G(x_1)^{-1})  \right], \\ 
\intertext{where we use our main assumption and the inequality $1-x \geq e^{-2x}$ for small $x$. Next, we observe that for all $x_0 < x < x_1$, $  (1-\epsilon) \ell(n \log^{-1} n)  \leq    \ell (i \delta n  \log^{-1} n)  \leq (1+\epsilon) \ell(n  \log^{-1} n), i = 1,2$. Thus,}
\Pr[N_\b > n]& \geq n^{-4 \delta} \left[ (1-\epsilon)  \ell(n \log^{-1} n ) \left( \frac{2 \delta \log n}{n} \right)^\alpha - (1+\epsilon)  \ell(n \log^{-1} n ) \left( \frac{ \delta \log n}{n} \right)^\alpha \right]  \\
&    \geq       n^{-4 \delta}  \frac{ \ell(n / \log n) \delta^\alpha  \log^\alpha n } {n^{\alpha} }  \left[2^ \alpha (1-\epsilon)^\alpha - (1+\epsilon)^ {\alpha}  \right]. 
\end{align*}
Next, since $\delta =  \epsilon \alpha / 4$, we have
\begin{align*}
\Pr[N_\b > n] & \geq   n^{-\alpha \epsilon}  \frac{ \ell(n / \log n) \delta^\alpha  \log^\alpha n } {n^{\alpha} }  h_\epsilon,  
\end{align*}
where we set $h_{\epsilon} = \delta^{\alpha} (2^\alpha (1-\epsilon)^\alpha - (1+\epsilon)^ {\alpha}) >0 $ since $\epsilon < 1/3$. Now, recalling the standard property of slowly varying functions (Theorem~1.5.6 on page 25 of \cite{BG87}) of $\ell(n)$ that $\ell(n) \geq n^{-\epsilon \alpha }$, for large $n$, we have
\begin{align*}
\Pr[N_\b > n]  & \geq  \frac{ h_{\epsilon}  \log^{\alpha(1+\epsilon)} n   } {n^{\alpha(1+2\epsilon) } }   \geq  \frac{h_{\epsilon} } {n^{\alpha(1+2\epsilon) } } ,
\end{align*}
since $\log n \geq 1$ for all $n>2$. 

And by taking the logarithm, we obtain 
 \begin{align*}
 \log \Pr[N_\b > n]& \geq \log h_{\epsilon} - \alpha (1+2\epsilon)\log n \\
 & \geq - \alpha(1+3\epsilon) \log n,
\end{align*}
since we can choose $n_0$ such that for all $n \geq n_0$, $\alpha \epsilon \log n \geq  -\log h_{\epsilon} $. Finally, dividing by $\alpha \log n > 0$, $n>1$, yields
\begin{align*}
 \frac{\log \Pr[N_\b > n]}{\alpha \log n} \geq  -(1+ 3\epsilon),
\end{align*}
which completes the proof after replacing $\epsilon$ with $\epsilon/3$.
\qed \end{proof}

\appendix
\section{A note on the uniform distribution of $F(L)$}
Here, we give a brief comment on the fact that $F(L)$ is uniformly distributed if $L$ is a continuous random variable, i.e., $F(x)$ is absolutely continuous. This statement is immediate when $F(x)$ is strictly increasing, e.g., Proposition 2.1 in Chapter 10 of \cite{Ross02}, since its inverse $F^{\leftarrow}(x)$ is well defined and $F(F^{\leftarrow}(x))=F^{\leftarrow}(F(x))=x$, meaning that, for $0 \leq x \leq 1$,  
\begin{align*}
\Pr[F(L) \leq x] = \Pr[L \leq F^{\leftarrow}(x) ]= F(F^{\leftarrow}(x)) = x.
\end{align*} 

However, when $F(x) $ has flat intervals, $F^{\leftarrow}(x)$ is not uniquely determined, e.g., we may use a standard generalized inverse: $F^{\leftarrow}(x) = \inf \{y : F(y) > x\}$. Now, we still have $F(F^{\leftarrow}(x))=x$, but $F^{\leftarrow}(F(x)) \neq x$, in general; the equality only holds if $x$ is a point of increase for $F(x)$. But, since $F(x)$ is absolutely continuous, we can exclude the flat regions of $F(x)$. Formally, the derivative $f(x) = F'(x)$ exists a.e.-Lebesgue, and where it does not, we can set $f(x) = 0$. Furthermore, $F(x) = \int_{-\infty}^x f(u) du$. Thus, since $F^{\leftarrow}(x)$ is strictly increasing, for $ 0 \leq x \leq 1$, 
\begin{align*}
\Pr[F(L) \leq x] &= \Pr[F^{\leftarrow} (F(L)) \leq F^{\leftarrow}(x) ] \\
&=  \int_{-\infty}^{\infty} \textbf{1}(F^{\leftarrow} (F(u)) \leq F^{\leftarrow}(x), f(u) >0 ) dF(u) = \Pr[L \leq F^{\leftarrow}(x)],
\end{align*} 
since $F(F^{\leftarrow}(u))=u$ when $f(u) >0$, implying that $\Pr[F(L)\leq x]=F(F^{\leftarrow}(x))=x$.

\end{document}